\documentclass[reqno]{amsart}

\usepackage{amssymb}

\usepackage{hyperref}
\usepackage[usenames,dvipsnames]{color}

\numberwithin{equation}{section}

\newcommand{\eq}[1]{\eqref{#1}}
\newcommand{\qtx}[1]{\quad\text{#1}\quad}

\newcommand{\beq}{\begin{equation}}
\newcommand{\eeq}{\end{equation}}

\newcommand{\abs}[1]{\left\lvert #1 \right\rvert}
\newcommand{\norm}[1]{\left\lVert #1 \right\rVert}

\newcommand{\set}[1]{\left\{ #1 \right\}}
\newcommand{\pa}[1]{\left( #1 \right)}

\DeclareMathOperator{\supp}{supp}  
\DeclareMathOperator{\im}{\mathrm{Im}}  
\DeclareMathOperator{\re}{\mathrm{Re}}

\newcommand{\ci}{\cite}

\newtheorem{theorem}{Theorem}[section]
\newtheorem{lemma}[theorem]{Lemma}
\newtheorem{cor}[theorem]{Corollary}
\newtheorem{prop}[theorem]{Proposition}
\newtheorem{remark}[theorem]{Remark} 
\newtheorem{defs}[theorem]{Definition}

\newcommand{\bl}{\begin{lemma}}
\newcommand{\el}{\end{lemma}}
\newcommand{\bc}{\begin{cor}}
\newcommand{\ec}{\end{cor}}
\newcommand{\bp}{\begin{prop}}
\newcommand{\ep}{\end{prop}}

\newtheorem*{assumptions}{Assumptions}

\newcommand{\de}{\delta}
\newcommand{\De}{\Delta}

\newcommand{\si}{\sigma}

\newcommand{\lb}{\lambda}
\newcommand{\ze}{\zeta}

\newcommand{\bvp}{{\boldsymbol{\varphi}}}

\newcommand{\ps}{\psi}

\newcommand{\BPhi}{\mathbf{\Phi}}
\newcommand{\BPsi}{\boldsymbol{\Psi}}
\newcommand{\La}{\Lambda}
\newcommand{\ve}{\varepsilon }

\newcommand{\e}{{\rm e}}

\newcommand{\Hh}{\mathcal{H}}
\newcommand{\F}{\mathcal{F}}
\newcommand{\K}{\mathcal{K}}
\newcommand{\T}{\mathcal{T}}
\newcommand{\cB}{\mathcal{B}}

\newcommand{\Vv}{\mathcal{V}}
\newcommand{\Ll}{\mathcal{L}}
\newcommand{\Tr}{{\rm Tr}}
\newcommand{\Mm}{\mathcal{M}}
\newcommand{\Cc}{\mathcal{C}}
\newcommand{\Uu}{\mathcal{U}}
\newcommand{\Pp}{\mathcal{P}}
\newcommand{\Ee}{\mathcal{E}}
\newcommand{\Ss}{\mathcal{S}}
\newcommand{\Ff}{\mathcal{F}}

\newcommand{\Jj}{\mathcal{J}}

\newcommand{\RR}{\mathbb{R}}
\newcommand{\NN}{\mathbb{N}}
\newcommand{\ZZ}{\mathbb{Z}}

\newcommand{\CC}{\mathbb{C}}

\newcommand{\B}{\mathbb{B}}

\newcommand{\E}{\mathbb{E}}
\newcommand{\G}{\mathbb{G}}

\newcommand{\aaa}{\mathbf{a}}
\newcommand{\bbb}{\mathbf{b}}
\newcommand{\ccc}{\mathbf{c}}

\newcommand{\Sym}{{\rm Sym}}
\newcommand{\diag}{{\rm diag}}
\newcommand{\sgn}{{\rm sgn}}

\newcommand{\hn}{|\!|\!|}

\newcommand{\hnn}{|\!|\!|\!|}

\newcommand{\PE}{\mbox{$\Pp\Ee$}}

\newcommand{\od}{\odot}
\newcommand{\ot}{{\od 2}}

\begin{document}

\title[A.C. Spectrum for Random Schr\"odinger Operators on the Bethe Strip]{Absolutely Continuous Spectrum for Random Schr\"odinger Operators on the Bethe Strip}

\author{Abel Klein}
\address[Klein]{University of California, Irvine,
Department of Mathematics,
Irvine, CA 92697-3875,  USA}
 \email{aklein@uci.edu}

\author{Christian Sadel}
\address[Sadel]{University of California, Irvine,
Department of Mathematics,
Irvine, CA 92697-3875,  USA}
 \email{csadel@math.uci.edu}

\thanks{A.K was  supported in part by the NSF under grant  DMS-1001509.}

%\subjclass{Primary 82B44, Secondary 47B80, 60H25}  
%\keywords{random Schr\"odinger operators, Anderson model, absolutely continuous spectrum, extended states}

\begin{abstract}
The Bethe Strip of width $m$ is the cartesian product $\B\times\{1,\ldots,m\}$, where $\B$ is the Bethe lattice  (Cayley tree). We prove that Anderson models  on the Bethe strip
have  ``extended states'' for small disorder. More precisely, we consider  Anderson-like Hamiltonians    $\;H_\lb=\frac12 \De \otimes 1 + 1 \otimes A\,+\,\lb \Vv$ on a Bethe strip with connectivity $K \ge 2$, where $A$ is an $m\times m$ symmetric matrix,  $\Vv$ is a random matrix potential, and $\lambda$ is the disorder parameter.  Given any closed interval $I\subset  (-\sqrt{K}+a_{\mathrm{max}},\sqrt{K}+a_{\mathrm{min}})$, where $a_{\mathrm{min}}$ and $a_{\mathrm{max}}$ are the smallest and largest eigenvalues of the matrix  $A$,  we prove that for $\lambda$ small  the random Schr\"odinger operator  $\;H_\lb$ has 
purely absolutely continuous spectrum in $I$ with probability one  and its  integrated density of
 states is continuously differentiable on the interval $I$. 
\end{abstract}

\maketitle 

%%%%%%%%%%%%%%%%%%%%%%%%%%%%%%%%%%%%%%%%%%%%%%%%%%%%%%%%%%%%%%%%%%%%%%%%%%%%%%%%%%%%%%%%
%%%%%%%%%%%%%%%%%%%%%%%%%%%%%%%%%%%%%%%%%%%%%%%%%%%%%%%%%%%%%%%%%%%%%%%%%%%%%%%%%%%%%%%%
%%%%%%%%%%%%%%%%%%%%%%%%%%%%%%%%%%%%%%%%%%%%%%%%%%%%%%%%%%%%%%%%%%%%%%%%%%%%%%%%%%%%%%%%
%%%%%%%%%%%%%%%%%%%%%%%%%%%%%%%%%%%%%%%%%%%%%%%%%%%%%%%%%%%%%%%%%%%%%%%%%%%%%%%%%%%%%%%%

\section{Introduction}

%%%%%%%%%%%%%%%%%%%%%%%%%%%%%%%%%%%%%%%%%%%%%%%%%%%%%%%%%%%%%%%%%%%%%%%%%%%%%%%%%%%%%%%%
%%%%%%%%%%%%%%%%%%%%%%%%%%%%%%%%%%%%%%%%%%%%%%%%%%%%%%%%%%%%%%%%%%%%%%%%%%%%%%%%%%%%%%%%
%%%%%%%%%%%%%%%%%%%%%%%%%%%%%%%%%%%%%%%%%%%%%%%%%%%%%%%%%%%%%%%%%%%%%%%%%%%%%%%%%%%%%%%%
%%%%%%%%%%%%%%%%%%%%%%%%%%%%%%%%%%%%%%%%%%%%%%%%%%%%%%%%%%%%%%%%%%%%%%%%%%%%%%%%%%%%%%%%

The Bethe strip of width $m$ is the cartesian product $\B\times\{1,\ldots,m\}$, where 
$\B$ denotes (the vertices of) the Bethe lattice (Cayley tree),  an infinite connected graph with no closed 
loops and a fixed degree (number of nearest neighbors) at each vertex.  
This fixed degree will be written as $K+1$ with $K\in \NN$ called the connectivity of $\B$.
The distance between two sites $x$ and $y$ {of $\B$} will be denoted by $d(x,y)$
and is equal to the length of the shortest path connecting $x$ and $y$.
The $\ell^2$ space of functions on the Bethe strip, $\ell^2(\B\times\{1,\ldots,m\})$, 
will be identified, as needed, with the tensor product
$\ell^2(\B)\otimes\CC^m$,  with the direct sum $\bigoplus_{x\in\B} \CC^m$, and with   $\ell^2(\B,\CC^m)=\big\{u:\B \mapsto \CC^m\,; \sum_{x\in\B} \|u(x)\|^2 < \infty \big\}$,
the space of  $\CC^m$-valued $\ell^2$ functions on $\B$,   {i.e.},
\begin{equation}
\ell^2(\B\times\{1,\ldots,m\})\;\cong\;\ell^2(\B)\otimes\CC^m\;\cong\;\bigoplus_{x\in\B} \CC^m\;\cong\; \ell^2(\B,\CC^m) \;.
\end{equation}

We consider the family of random Hamiltonians on $\ell^2(\B\times\{1,\ldots,m\})$ given by
\begin{equation}\label{Hlambda}
H_\lb\; = \tfrac12\;\De\otimes 1\,+\,1\otimes A\,+\,\lb\Vv \;.
\end{equation}
Here $\De$ denotes the centered Laplacian on $\ell^2(\B)$, which has 
spectrum $ \sigma(\Delta)  =  [-2\sqrt{ K},2\sqrt{K}]\,$ (e.g., \ci{AK}).  We use $\frac{1}{2}\De$
in the definition of $H_\lb$ to simplify some formulas.
$A \in \Sym(m)$ denotes the ``free vertical operator'' on the Bethe strip, where $\Sym(m)\cong \RR^{\frac 1 2 {m(m+1)}}$ is the set of
real symmetric $m\times m$ matrices.
$\Vv$ is the random matrix-potential given by $\Vv=\bigoplus_{x\in\B} V(x)$ on $\bigoplus_{x\in\B} \CC^m$, where 
$\{V(x)\}_{x\in\B}$ are independent identically distributed $\Sym(m)$-valued random variables with common probability distribution $\mu$.
The coefficient  $\lb$ is a real parameter called the {\em disorder}.
In particular, for  $u\in\ell^2(\B,\CC^m)$ we have 
\begin{equation}
(H_\lb u)(x) = \tfrac 1 2   \!\!\!\!\!
\sum_{\substack{y\in \B\\  d(x,y)=1}}  \!\!\!\!{u(y)}\;+\;A\,u(x)\,+\,\lambda\, V(x)\, u(x)\quad\text{for all}\quad x \in \B\;.
\end{equation}

An important special case of this model is the Anderson model on the product graph $\B\times\G$, where
$\G$ is a finite graph with $m$ labeled vertices.  If $ A_\G$ is the adjacency matrix of the graph $\G$, i.e.,
$(A_\G)_{k,\ell}$ denotes the number of edges between $k\in\G$ and $\ell \in \G$, then
$\De\otimes 1\,+\,1\otimes A_\G$ is the {adjacency} operator on the {product graph} $\B \times \G$. {If in \eqref{Hlambda}
  we take
$A=\frac12 A_\G$ and $\mu$   supported 
by the diagonal matrices, with the diagonal entries being independent identically distributed,  then $H_\lb$
is} the Anderson model on the product graph $\B\times\G$.  {The Anderson model is a particular case of the matrix Anderson model:  $H_{\lambda}$ as in \eqref{Hlambda} with $A=\frac12 A_\G$.

The Anderson model \ci{And}  describes} the motion of a quantum-mechanical
electron in a crystal with impurities.  {If $\lb \not= 0$, the following picture is widely accepted \cite{And,AALR}: In one and two dimensions  the Anderson model always  exhibits  localization (i.e.,
 pure point spectrum with exponentially decaying eigenfunctions). In three and more dimensions 
both localized and extended states (i.e., absolutely continuous spectrum) are expected for small
disorder, with the energies of extended and localized  states being separated by a ``mobility edge''.

We  have by now a good understanding of localization.  For the Anderson model there is  always localization in dimension $d=1$ \cite{KS1,CKM} and on the one-dimensional strip \cite{Lac,KLS}. For dimensions $d \ge 2$, with suitable regularity conditions on the single site probability distribution there is always  localization 
at high disorder or at the edges of the spectrum \cite{FMSS,DLS,SW,CKM,DK,K3,A,AM,Wang,Klo}. The expected localization at all disorders in dimension $d=2$ remains an open problem.

On the other hand, there are no results on  the expected existence of absolutely continuous spectrum for the Anderson model in dimension $d=3$ or higher.  Existence of absolutely continuous spectrum has only been proven for the Anderson model
on the Bethe lattice and similar tree like structures.  Klein proved that, at low disorder, the Anderson model on the Bethe lattice
  has purely absolutely
continuous spectrum in a nontrivial  interval  \cite{K1,K2,K9}     and exhibits 
ballistic behavior \cite{K8}. More recently, different proofs for the existence of absolutely
continuous spectrum on the Bethe lattice and similar tree structures have been  provided in
\cite{ASW,FHH,FHS,H,KLW}.  Absolutely continuous spectrum has also been shown in   models were certain symmetries prevent localization, e.g., \cite{SS}.

 Recently,  Froese, Halasan and Hasler  \cite{FHH} extended the hyperbolic geometry methods used in \cite{FHS,H} to an Anderson model
on the Bethe strip  with connectivity $K=2$ and width $m=2$, proving the existence of absolutely continuous spectrum in an interval at low disorder.  Their method requires  working in  the Siegel upper half plane when $m=2$ instead of working in the upper half plane as when $m=1$.  They also conjectured the analogous result for general Bethe strips.

Klein's original proof of absolutely continuous spectrum for the Anderson model on the Bethe lattice \cite{K1,K2,K9}  relied on the Implicit Function Theorem on Banach spaces and 
some crucial identities arising from a supersymmetric formalism. These ideas are extended to the Bethe strip in this article. In particular,  we  prove the conjecture  in  \cite{FHH}, providing an extension of their results  to Bethe strips of arbitrary connectivity $K\ge 2$ and width $m\in \NN$.

In a sequel to this paper  we prove  ballistic behavior  for the Anderson model in the Bethe strip at low disorder \cite{KS},  extending the results of \cite{K8}.

Going to the strip requires an extension of the supersymmetric formalism, as already seen by Klein  and Speis \cite{KSp2,KLS2} in the one-dimensional strip.  The change is akin to going from one-variable to multi-variable calculus. The formalism becomes more cumbersome: scalar quantities are replaced by matrix quantities, derivatives are replaced by partial derivatives, etc.
In particular, 
 a  difference appears between Bethe strips with width $m=1,2$ and those with bigger widths, i.e., $m=3,4,\ldots$.  If  $m=1,2$ only one 
replica of the supersymmetric variables suffices. But if 
 $m\geq 3$ one needs  $n\ge \frac m 2$  replicas,  as noted in \cite{KSp2}. This leads to more complicated function spaces, the fixed point analysis that is the crux of the proof is conducted on function spaces requiring derivatives up to  order $nm \ge \frac {m^2} 2$, not just of order $m$.  To use the Implicit Function Theorem, one needs to prove the invertibility of  certain operators in these spaces.  This was done in \cite{K2} (for $m=1$) by calculating eigenvalues and eigenfunctions of the operators explicitly, and  proving that the linear span of these eigenfunctions is dense in the relevant Banach space.
 This density argument, which relied on  results of Acosta and Klein \cite{AK},  does not carry over to the case $m\ge 2$.  In this article  we use a different approach, showing that it suffices to carry the analysis in  function spaces defined as  closures of the  linear span of certain  eigenfunctions.
 
Besides the Anderson model on $\B\times\G$, another interesting special case of  \eqref{Hlambda} is the Wegner $m$-orbital model on the Bethe lattice: Set
$A=0$ and  let the random matrix $V(x)$ be distributed as in the Gaussian Orthogonal Ensemble  (GOE).  Then  $V(x)=V(x)^t$  and the upper triangular  entries
are independent, centered, Gaussian variables with variance 1 along the diagonal and variance $\frac12$ for the off-diagonal entries.
This model was  introduced by Wegner \cite{Weg} on the lattice   $\ZZ^d$, where he studied   the limit $m\to\infty$.  
Dorokhov \cite{Dor} studied a related quasi one-dimensional model.  If $A=0$,
we will call $H_{\lambda}$ in \eqref{Hlambda} a general Wegner $m$-orbital model .

%%%%%%%%%%%%%%%%%%%%%%%%%%%%%%%%%%%%%%%%%%%%%%%%%%%%%%%%
%%%%%%%%%%%%%%%%%%%%%%%%%%%%%%%%%%%%%%%%%%%%%%%%%%%%%%%%5

Analogous to the case of the Bethe lattice, it follows from ergodicity (the ergodic theorem in the Bethe lattice is discussed in 
 \ci[Appendix]{AK}) that the spectrum of the Hamiltonian $\, H_\lb$ is given by
\begin{equation} \label{eq-spectrum}
\sigma(H_\lb) = \sigma(\tfrac{1}{2}\Delta) + \bigcup_{V\in \supp \mu} \sigma(A+\lambda V) 
\end{equation}
with probability one \ci{P,CL}, where $\sigma(A+\lambda V)$ denotes the set of eigenvalues of the {$m \times m$} matrix $A+\lambda V$.
For each choice of $\Vv$ the spectrum of $H_\lb$ can be decomposed
into pure point spectrum, $\sigma_{pp}(H_\lb)$, absolutely continuous
spectrum, $\sigma_{ac}(H_\lb)$, and singular continuous spectrum,
$\sigma_{sc}(H_\lb)$.  Ergodicity gives the existence of sets $\Sigma_{\lb,pp}\, , \; \Sigma_{\lb,ac}\, ,
\; \Sigma_{\lb,sc} \subset {{\RR}}$,  such that $\sigma_{pp}(H_\lb) = \Sigma_{\lb,pp}
\; , \; \sigma_{ac}(H_\lb) = \Sigma_{\lb,ac}$, $\sigma_{sc}(H_\lb) = \Sigma_{\lb,sc}$
with probability one \ci{KS1,CL}.

\begin{assumptions} In this article we always make the following assumptions:
\begin{enumerate}
\item[(I)]  $K \ge 2$,  so $\B$ is not equal to $\ZZ$. 
\item[(II)] The common probability distribution  $\mu$ of the $\Sym(m)$-valued random variables $\{V(x)\}_{x\in\B}$ has finite  (mixed) moments of all orders. In particular, the characteristic function
of $\mu$,
\begin{equation}
h(M) : = \int_{{\Sym(m)}} {{\e}^{-i\Tr(MV)}d\mu (V)}\quad\text{for}\; M\, \in\,\Sym(m)\;, 
\end{equation}
is a $C^{\infty}$ function on $\Sym(m)$ with bounded derivatives.

\item[(III)] Let   $a_{\mathrm{min}} :=a_1\leq a_2\leq \ldots \leq a_m= : a_{\mathrm{max}}$ be the eigenvalues of the  ``free vertical operator'' $A$, and set
\begin{equation} \label{IAK}
I_{A,K} = \bigcap_{i=1}^n (-\sqrt{K}+a_i,\sqrt{K}+a_i) = (-\sqrt{K}+a_{\mathrm{max}},\sqrt{K}+a_{\mathrm{min}}).
\end{equation}
The interval $I_{A,K}$ is not empty, i.e.,
\begin{equation}\label{IAKhyp}
a_{\mathrm{max}}-a_{\mathrm{min}} < 2\sqrt{K}.
\end{equation}  

\end{enumerate}
\end{assumptions}

It would suffice to require $\mu$ to have finite  moments up to order $2m \left\lceil\frac m 2 \right\rceil$; we require all moments for simplicity.
Note also that for a fixed $A$ we can always obtain \eq{IAKhyp}  by taking $K$ big enough.

If $I\subset I_{A,K}$ is a  compact interval, it follows from  \eqref{eq-spectrum}  that $I\subset\sigma(H_\lb)$ for $\lambda$ small enough.
We will prove that under the above assumptions $H_\lambda$ has 
``extended states'' in $I$ for small disorder.

\begin{theorem} \label{main} For any compact interval $I\subset I_{A,K}  $ 
there exists $\lb(I) > 0$, such that for any $\lb$ with
$ |\lb| < \lb(I) $ the spectrum of $\;H_\lb$ in $ I$  is 
purely absolutely continuous with probability one, i.e., 
 we have $\Sigma_{\lb,ac} \cap I =  I$ and  
$\Sigma_{\lb,pp} \cap I = \Sigma_{\lb,sc} \cap I=\emptyset $.
\end{theorem}

In particular we get the following interesting cases.

\begin{cor} 
Let $H_\lb$ be a matrix Anderson model {(i.e., $A=\frac12 A_\G$)} on $\B\times \G_m$,  $m\in \{2,3,4,\ldots\}$, where
$\G_2$ denotes the finite graph with two vertices and one edge connecting them, and
for $m\ge 3$, 
$\G_m$ denotes the finite loop with $m$ vertices where every vertex is connected to two neighbors. Let the open intervals  $I_m$ be defined by
\beq
I_m =\begin{cases} \big(-\sqrt{K}+\tfrac 1 2\,,\, \sqrt{K}-\tfrac 1 2\big) & \text{if}\quad  m=2\\
 \big(-\sqrt{K}+1\,,\,\sqrt{K}+\cos (\tfrac {m-1}m\pi) \big) & \text{if}\quad  m=3,5,7,\ldots\\
 \big(-\sqrt{K}+1\,,\,\sqrt{K}-1 \big) & \text{if}\quad  m=4,6,8,\ldots
\end{cases}.
\eeq
Then for all compact intervals $I\subset I_m$, if $\lambda$ is small enough,  the 
matrix Anderson model  $H_\lb$ has {purely absolutely continuous spectrum in $I$ with probability one.}
\end{cor}

\begin{cor}
 Let $H_\lb$ be a {general} Wegner $m$-orbital model on the Bethe lattice {(i.e., $A=0$)}.
 Then for all compact intervals $I\subset (-\sqrt{K}\,,\,\sqrt{K})$, if $\lambda$ is small enough,
the general Wegner $m$-orbital model  $H_\lb$ has {purely absolutely continuous spectrum in $I$ with probability one.}
\end{cor}

The key object to be analyzed is the $m\times m$  matrix Green's function of $H_\lb$:   
\begin{equation}
\left[G_{\lb}\, (x,y;z)\right]_{j,k}\; = \;\left\langle {x,j|(H_{\lb} -z)^{-1}|y,k} \right\rangle, 
\end{equation}
where $x,y \in {\B}$, $j,k \in\{1,\ldots,m\}$, and $z = E + i\eta$ with $E \in \RR$, $\eta> 0$.  Here $|x,k\rangle$ denotes the $\CC^m$-valued function $u(y)=\delta_{x,y} e_k$, where $e_k$ is the $k$-th canonical basis vector in $\CC^m$,
so $\{|x,k\rangle; \; x\in\B,\,k\in\{1,\ldots,m\}\}$ is an orthonormal basis for $\ell^2(\B,\CC^m)$.
{Similarly} to the Bethe lattice (see \ci{AK} for a discussion of the integrated density of states in the Bethe lattice) 
we define the integrated density of states $N_\lb (E)$ by
\begin{equation}
N_\lb (E)\; = \;\tfrac1m\;\E\left(\sum_{k=1}^m\left\langle {x,k|\chi_{(-\infty,E]}(H_{\lb})|x,k} \right\rangle\right)\;\;\;
\mbox{for any}\;\; x \in \B\;, 
\end{equation}
where $\E$ denotes the expectation with respect to the distribution of $\{V(x)\}_{x\in\B}$.

 For any  $x \in \B$ and any potential $\Vv$,
$  G_{\lb}\, (x,x; E +i\eta) $ is a continuous function of  
$  (\lb, E,\eta) \in \RR \times \RR \times (0, \infty)$; to prove it one uses the  
resolvent identity plus the fact that,  as long as $\eta > 0$, we have  
$ \lb \Vv(\lb \Vv - i\eta)^{-1}  \to 0$ strongly as $\lb \to 0$.
 It then follows from the Dominated Convergence Theorem  that 
 $\E (G_{\lb}\, (x,x; E +i\eta) )$ and $ \E (|G_{\lb}\, (x,x; E +i\eta)|^2) $ are also
continuous functions of  $  (\lb, E,\eta) \in \RR \times \RR \times (0, \infty)$. The crucial observation  is that we can let $\eta \downarrow 0$  inside $I_{A,K}$.

\begin{theorem} \label{cor} For any compact interval $I\subset I_{A,K}$ there exists $\lb(I) > 0$ such that:
\begin{enumerate} 
\item[(i)] For all $x \in \B$ 
the continuous functions 
\begin{align}
(\lb, E,\eta) &\in (-\lb(I),\lb(I)) \times I \times (0, \infty) \;\longrightarrow\;
 \E (G_{\lb}\, (x,x; E 
  +i\eta) ),\\
(\lb, E,\eta) &\in (-\lb(I),\lb(I)) \times I \times (0, \infty) \;\longrightarrow\;
 \E (|G_{\lb}\, (x,x; E +i\eta)|^2  )
\end{align}
 have continuous extensions to  $ (-\lb(I),\lb(I)) \times I \times [0, \infty) $.

 \item[(ii)] For any $\lb$ with
$ |\lb| < \lb(I) $ the integrated density of states $N_\lb (E)$ is continuously differentiable
on $ \mathring{I}$, the interior of  ${I}$,  and for all $E \in  \mathring{I}$ we have 
\beq
N'_\lb (E) =\lim_{\eta \downarrow 0} \frac{1}{\pi} \mbox{\em Im}\, \E \Tr(G_{\lb}\, (x,x; E +i\eta) ) \qtx{for any} x \in \B.
\eeq

 \item[(iii)] For all $x \in \B$ we
have 
\begin{equation}
\sup_{\lb;\, |\lb| < \lb(I)}\,  \sup_{E \in I} \, \sup_{\eta;\, 0 <\eta }\,
\E (\Tr(|G_{\lb}\, (x,x; E +i\eta)|^2)  )\;\;<\;\; \infty .     \label{sup}
\end{equation}
 \end{enumerate}
\end{theorem}

Theorem \ref{main}  will follow from part (iii) in the theorem, an immediate consequence of  part (i).  Realizing that
$ \frac1m\,\E\, \Tr (G_{\lb}\, (x,x; E +i\eta) )$ is 
the Stieltjes transform of the integrated density of states, we see that
part (ii) also  follows from part (i). Thus, we only need to prove part (i).

Since $A$ is symmetric, there exists an orthogonal matrix $O$ such that $O^t A O$ is diagonal. Then $\Uu=1\otimes O$ is a unitary transformation on
 $\ell^2(\B,\CC^m)$, such that
$\Uu^*H_\lambda \Uu = \frac12 \Delta\otimes 1 \,+\,1\otimes O^t A O \,+\,\lambda \Uu^* \Vv\Uu$.
Now $\Uu^*\Vv\Uu = \bigoplus_{x\in\B} O^tV(x)O$ is a matrix potential like $\Vv$ itself. 
Hence by conjugating the distribution $\mu$ of the matrix potential $V(x)$
we can assume, without loss of generality,  that $A$ is diagonal and we will do so in the proofs.
Thus from now on we assume 
\begin{equation}
 A = \diag\,(a_1,\ldots, a_m)\;, \qtx{i.e.,} A_{j,k}=\delta_{j,k}a_k \qtx{for} j,k=1,2,\dots,m.
\end{equation}

This article is organized as follows:  In Section~\ref{sec-super}  we   introduce the supersymmetric formalism
and the crucial supersymmetric function spaces. In Section~\ref{sec-avgreen} we use the supersymmetric replica trick to rewrite the matrix Green's function, and derive a fixed point equation for a certain supersymmetric function from which we calculate the averaged Green's matrix. In Section~\ref{sec-avsqgreen} we obtain analogous results for the averaged squared matrix Green's function.
In Section~\ref{sec-fxp} we  perform a fixed point analysis using the Implicit Function Theorem
to show the existence of continuous extensions for the solutions of the fixed point equation to energies on the real line.  Finally, in Section~\ref{sec-proofs}  we show that these continuous extensions   yield the proofs of Theorems~\ref{main} and \ref{cor}.

%%%%%%%%%%%%%%%%%%%%%%%%%%%%%%%%%%%%%%%%%%%%%%%%%%%%%%%%%%%%%%%%%%%%%%%%%%%%%%%%%%%%%%%%
%%%%%%%%%%%%%%%%%%%%%%%%%%%%%%%%%%%%%%%%%%%%%%%%%%%%%%%%%%%%%%%%%%%%%%%%%%%%%%%%%%%%%%%%
%%%%%%%%%%%%%%%%%%%%%%%%%%%%%%%%%%%%%%%%%%%%%%%%%%%%%%%%%%%%%%%%%%%%%%%%%%%%%%%%%%%%%%%%
%%%%%%%%%%%%%%%%%%%%%%%%%%%%%%%%%%%%%%%%%%%%%%%%%%%%%%%%%%%%%%%%%%%%%%%%%%%%%%%%%%%%%%%%

\section{The supersymmetric formalism \label{sec-super}}

%%%%%%%%%%%%%%%%%%%%%%%%%%%%%%%%%%%%%%%%%%%%%%%%%%%%%%%%%%%%%%%%%%%%%%%%%%%%%%%%%%%%%%%%
%%%%%%%%%%%%%%%%%%%%%%%%%%%%%%%%%%%%%%%%%%%%%%%%%%%%%%%%%%%%%%%%%%%%%%%%%%%%%%%%%%%%%%%%
%%%%%%%%%%%%%%%%%%%%%%%%%%%%%%%%%%%%%%%%%%%%%%%%%%%%%%%%%%%%%%%%%%%%%%%%%%%%%%%%%%%%%%%%
%%%%%%%%%%%%%%%%%%%%%%%%%%%%%%%%%%%%%%%%%%%%%%%%%%%%%%%%%%%%%%%%%%%%%%%%%%%%%%%%%%%%%%%%

The formalism described here can be found in more detail in
\cite{B,E,K,KSp2}. In particular,  \cite{KSp2} contains all the  important formulas we need.   In this section we introduce our notation and  review the relevant results.

\subsection{Supervariables and supermatrices}

Given  $m,n \in \NN$, let $\{\psi_{k,\ell}, \overline{\psi}_{k,\ell}; \; k=1,\ldots,m,\,\ell=1,\ldots,n \}$  be
$2mn$ independent Grassmann variables, { i.e.,} they are generators of a
Grassmann algebra isomorphic to $\Lambda^{2mn}(\RR)$. In particular, they all anti-commute and the 
algebra is given by the free algebra over $\RR$ generated by these symbols modulo the ideal generated by the anti-commutators 
\[
\psi_{i,j} \psi_{k,\ell} + \psi_{k,\ell} \psi_{i,j} , \quad  \overline{\psi}_{i,j} \overline{\psi}_{k,\ell} + \overline{\psi}_{k,\ell} \overline{\psi}_{i,j},  \quad \overline{\psi}_{i,j} {\psi}_{k,\ell} + {\psi}_{k,\ell} \overline{\psi}_{i,j} ,
\]
where $ i,k=1,\ldots,m$ and $j,\ell=1,\ldots,n$.
This algebra is finite dimensional and will be denoted by $\Lambda(\BPsi)$, where 
$\BPsi$ denotes the matrix of pairs $\BPsi=(\overline{\psi}_{k,\ell},\psi_{k,\ell})_{k,\ell}$.
Its complexification will be denoted by
$\Lambda_\CC(\BPsi)=\CC\otimes_\RR \Lambda(\BPsi)$.
We denote the set of one forms (linear combinations of the generators) by $\Lambda^1(\BPsi)$.
Sometimes we will also allow to add and multiply expressions from different Grassmann algebras
$\Lambda(\BPsi)$ and $\Lambda(\BPsi')$. In this case, sums and products have to be understood in the Grassmann algebra
 $\Lambda(\BPsi,\BPsi')$, which is generated by the entries of $\BPsi$ and $\BPsi'$ as independent
Grassmann variables.

A supervariable is an element of $\RR^2\oplus\Lambda^1(\BPsi)\oplus\Lambda^1(\BPsi)$.
We introduce variables $\varphi_{k,\ell} \in \RR^2$ and consider the supervariables  
$\phi_{k,\ell}=(\varphi_{k,\ell},\overline{\psi}_{k,\ell},\psi_{k,\ell})$.
The collection $\BPhi=(\phi_{k,\ell})_{k,\ell}$
will be called a $m \times n$ supermatrix.
More generally, an $m\times n$ matrix 
$\tilde\BPhi=(\tilde\varphi_{k,\ell},\overline{\tilde{\psi}}_{k,\ell},\tilde\psi_{k,\ell})_{k,\ell}\,\in \,
\left[\RR^2\oplus\Lambda^1(\BPsi)\oplus\Lambda^1(\BPsi)\right]^{m\times n}$
will be called a supermatrix  if all the 
appearing one-forms $\overline{\tilde{\psi}}_{k,\ell},\tilde\psi_{k,\ell}$, $k=1,2\ldots,m$ and $\ell=1,2\ldots,n$, are linearly independent.
The collection of all supermatrices is a dense open subset of the vector space
$\left[\RR^2\oplus\Lambda^1(\BPsi)\oplus\Lambda^1(\BPsi)\right]^{m\times n}$
 and will be denoted by $\Ll_{m,n}(\BPsi)$, or just $\Ll_{m,n}$.
Linear maps defined on $\Ll_{m,n}(\BPsi)$ have to be understood as restrictions of linear maps 
defined on $\left[\RR^2\oplus\Lambda^1(\BPsi)\oplus\Lambda^1(\BPsi)\right]^{m\times n}$.

Supermatrices $(\BPhi_i)_i$ are said to be independent if
$\BPhi_i\in\Ll_{m,n}(\BPsi_i)$ for all $i$, and all the entries of the different
$\BPsi_i$ are independent Grassmann variables.

We also consider  matrices $\bvp=(\varphi_{k,\ell})_{k,\ell}$ with entries in $\RR^{2}$. Writing each entry
$\varphi_{k,\ell}$   as a row vector, $\bvp$ may be considered
as $m\times 2n$ matrix {with real entries}. 
Similarly, one may consider $\BPsi$ as $m\times 2n$ matrix with entries in $\Lambda^1(\BPsi)$.
With all these notations one may write $\BPhi=(\bvp,\BPsi)$,  splitting a
supermatrix into its real and Grassmann-variables parts.

For supervariables $\phi_1=(\varphi_1,\overline{\psi}_1,\psi_1)$ and $\phi_2=(\varphi_2,\overline{\psi}_2,\psi_2)$
we define
\begin{equation}
 \label{eq-dot-supervec}
\phi_1 \cdot \phi_2 : = \varphi_1\cdot\varphi_2+\tfrac12(\overline{\psi}_1 \psi_2 + \overline\psi_2 \psi_1)\;.
\end{equation}
By $\Phi_k$ we denote the $k$-th row vector $(\phi_{k,\ell})_{\ell=1\,\ldots,n}$ of {a supermatrix}  $\BPhi$.
For the row vectors of two supermatrices $\BPhi$ and $\BPhi'$ we set
\begin{equation} \label{eq-dot-rowvec}
\Phi'_j \cdot \Phi_k : = \sum_{\ell=1}^n \phi'_{j,\ell} \cdot \phi_{k,\ell}.
\end{equation}
We also define a dot product between supermatrices by
\begin{equation}\label{eq-dot-supermatrix}
 \BPhi' \cdot \BPhi : = 
\sum_{k=1}^m \Phi_k' \cdot \Phi_k = 
\sum_{k=1}^m \sum_{\ell=1}^n \phi'_{k,\ell}\cdot\phi_{k,\ell}\;.
\end{equation}
Furthermore, we introduce the $m\times m$ matrix $\BPhi^\ot$ with entries in $\Lambda(\BPsi)$ by
\begin{equation}\label{eq-supermatrix-tensor}
(\BPhi^{\od 2})_{j,k} :=  \Phi_j \cdot \Phi_k = \sum_{\ell=1}^n \phi_{j,\ell} \cdot \phi_{k,\ell} \;.
\end{equation}
In addition, for any complex $m\times m$ matrix $B$ and supermatrices $\BPhi, \,\BPhi'$ we set
\begin{equation}\label{eq-supermatrix-prod}
\BPhi'\cdot B \BPhi := \sum_{j,k=1}^m B_{j,k} \Phi'_j \cdot \Phi_k\;\in\;\Lambda_\CC(\BPsi)\;.
\end{equation}
Note that
\begin{equation}
\BPhi\cdot B \BPhi = \Tr(B \BPhi^\ot)\;.
\end{equation}

These definitions may be memorized as follows: If $n=1$, $\BPhi$ is a column vector indexed by $k$ 
and $B\BPhi$ corresponds to a matrix vector product and $\BPhi'\cdot B\BPhi$ is the dot product of vectors of supervariables.
For general $n$ the supermatrix  $\BPhi$ has columns indexed by $\ell=1,2,\ldots,n$, ``the $n$ replicas",  and in all definitions 
of dot products
there is an additional sum over this index.

For a supermatrix $\BPhi=(\bvp,\BPsi)$, 
where $\bvp\in\RR^{m\times 2n}$ and 
$\BPsi \in \Lambda^1(\BPsi)^{m\times 2n}$, one has  
\begin{align}
(\BPhi^\ot)_{j,k} &=
\sum_{\ell=1}^n \left\{\varphi_{j,\ell} \cdot \varphi_{k,\ell}+\tfrac12(\overline \psi_{j,\ell}\psi_{k,\ell}+\overline \psi_{k,\ell}\psi_{j,\ell}) \right\}
\label{eq-BPhi1}\\
\nonumber 
&=\sum_{\ell=1}^n \left\{\varphi_{j,\ell} \cdot \varphi_{k,\ell}+ \left[ \begin{matrix} \,\overline{\psi}_{j,\ell} & \psi_{j,\ell} \end{matrix}\right]
\left[\begin{matrix} \,0 & \tfrac12 \\ -\tfrac12 & 0 \end{matrix}\right]
\left[\begin{matrix} \overline{\psi}_{k,\ell} \\ \psi_{k,\ell} \end{matrix}\right]\right\}\;.
\end{align}
It follows that
\beq \label{eq-BPhi2}
\BPhi^\ot = \bvp^\ot\,+\,\BPsi^\ot\,, \qtx{with}
\bvp^\ot:=\bvp\bvp^t\;\text{and}\;  \BPsi^\ot:=\BPsi J \BPsi^t\; ,
\eeq
where $J$ is the $2n\times 2n$ matrix consisting of $n$ blocks 
$\left[\begin{smallmatrix} \,0 & \frac12 \\ -\frac12 & 0 \end{smallmatrix}\right]$ along the diagonal.

Given a matrix $B$ as in \eq{eq-supermatrix-prod} and $\bvp', \bvp\in \RR^{m\times 2n}$, we write
\begin{equation}\label{eq-supermatrix-prod25}
\bvp'\cdot B \bvp := \sum_{j,k,\ell} B_{j,k} \varphi'_{j,\ell} \cdot \varphi_{k,\ell} = \Tr((\bvp')^t B \bvp )\;\in\;\CC\; .
\end{equation}

A complex superfunction with respect to $\Lambda(\BPsi)$ is a function 
$F:\RR^{m\times 2n}\to \Lambda_\CC(\BPsi)$.
Let $\beta_{i} \in \Lambda(\BPsi)$ for $i\in\{1,\ldots,2^{2mn}\}$ 
be a basis {for} $\Lambda(\BPsi)$ over $\RR$. 
Each $\beta_{i}$ is a polynomial in the entries of $\BPsi$ (since we required the entries of $\BPsi$ to be independent)
and $F$ is of the form  
\begin{equation} \label{eq-superfct-expand}
F(\bvp) = \sum_{i=1}^{2^{2mn}}\, F_i(\bvp)\, \beta_i\;,\quad\text{where}\quad F_i\,:\,\RR^{m\times 2n}\,\to\,\CC\;.
\end{equation} 
We interpret this as a function $F(\BPhi)$ where 
$\BPhi=(\bvp,\BPsi)$. In this sense the map $\BPhi\mapsto \BPhi\cdot B\BPhi$ as in \eqref{eq-supermatrix-prod} can be considered as a 
superfunction.
Similarly one can define superfunctions $F(\BPhi_1,\ldots,\BPhi_k)$, of $k$  independent supermatrices,
using the Grassmann algebra $\Lambda((\BPsi_j)_{j\in\{1,\ldots,k\}})$.

We write $F\in C(\Ll_{m,n})$ and $F\in C^k(\Ll_{m,n})$ if 
$F_i\in C(\RR^{m\times 2n})$ or $F_i\in C^k(\RR^{m\times 2n})$ respectively,
for all $i$ in the expansion \eqref{eq-superfct-expand}.
Similarly, we write $F\in\Ss(\Ll_{m,n})$ if $F_i\in\Ss(\RR^{m\times 2n})$, the Schwartz space.

Let us now define the integral over the Grassmann variables in the following way.
For a fixed pair $k,\ell$ we  write $F=F(\BPhi)$ as  $F=F^{k,\ell}_0+F^{k,\ell}_1\overline{\psi}_{k,\ell}+F^{k,\ell}_2\psi_{k,\ell}+
F^{k,\ell}_3\overline{\psi}_{k,\ell}\psi_{k,\ell}$
where the $F^{k,\ell}_i$ are superfunctions not depending on $\overline{\psi}_{k,\ell}$ and $\psi_{k,\ell}$.
Then
\beq
\int F\,  d\overline{\psi}_{k,\ell}\,d\psi_{k,\ell}   := 
   -F^{k,\ell}_3 \;.
\eeq
If all functions $F_i$ in the expansion \eqref{eq-superfct-expand} are in $L^1(\RR^{m\times2n})$,  we say that
$F\in L^1(\Ll_{m,n})$ and define the supersymmetric integral by
\begin{equation}\label{eq-def-DPhi}
\int F(\BPhi)\;D\BPhi = \frac 1 {\pi^{mn}}
\int\,F(\BPhi)\; \prod_{k=1}^m \prod_{\ell=1}^{n} d^2 \varphi_{k,\ell}\;d\overline{\psi}_{k,\ell}\,d\psi_{k,\ell}\;.
\end{equation}

%%%%%%%%%%%%%%%%%%%%%%%%%%%%%%%%%%%%%%%%%%%%%%%%%%%%%%%%%%%%%%%%%%%%%%%%%%%%%%%%%%%%%%%%
%%%%%%%%%%%%%%%%%%%%%%%%%%%%%%%%%%%%%%%%%%%%%%%%%%%%%%%%%%%%%%%%%%%%%%%%%%%%%%%%%%%%%%%%
%%%%%%%%%%%%%%%%%%%%%%%%%%%%%%%%%%%%%%%%%%%%%%%%%%%%%%%%%%%%%%%%%%%%%%%%%%%%%%%%%%%%%%%%
%%%%%%%%%%%%%%%%%%%%%%%%%%%%%%%%%%%%%%%%%%%%%%%%%%%%%%%%%%%%%%%%%%%%%%%%%%%%%%%%%%%%%%%%

\subsection{Supersymmetries and supersymmetric functions}

To obtain the full set of supersymmetries,  we  introduce another Grassmann variable as in \cite{KSp2}.
So let $\xi$ be a new Grassmann variable, independent of $\BPsi=(\overline{\psi}_{k,\ell},\psi_{k,\ell})_{k,\ell}$,  $k=1,\ldots,m$ and $\ell=1,\ldots,n$, and let $\Lambda(\xi)$ and
$\Lambda(\BPsi,\xi)$ denote the Grassmann algebras generated by $\xi$ and {$\BPsi\cup \{\xi\}$,} respectively.
We consider the real vector space 
\beq
\Mm=\big(\RR \oplus \xi\Lambda^1(\BPsi)\big)^2=\RR^2\oplus[\xi\Lambda^1(\BPsi)]^2 \subset \Lambda(\BPsi,\xi)^2 .  
\eeq   
The scalar product on $\RR^2$ extends to a $\RR \oplus \xi\Lambda^1(\BPsi)$-valued scalar product on   $\Mm$ by
\beq
(a_1 + \xi \psi_1,a_2 + \xi \psi_2)\cdot(a_1^{\prime} + \xi \psi_1^{\prime} ,a_2^{\prime}  + \xi \psi_2^{\prime} )= a_1a_1^{\prime}  +  a_2a_2^{\prime}+ \xi (a_1\psi_1^{\prime} +a_1^{\prime}\psi_1 +a_2\psi_2^{\prime} +a_2^{\prime}\psi_2) 
\eeq
for $a_1,a_1^{\prime},a_2,a_2^{\prime}\in \RR$, $\psi_1,\psi_1^{\prime},\psi_2,\psi_2^{\prime}\in \Lambda^1(\BPsi)$.

A generalized {supervariable} is a triple $(\varphi, \overline{\psi},\psi)$ where $\varphi\in\Mm$ 
and $\overline{\psi},\psi \in \Lambda^1(\BPsi,\xi)$. Generalized supermatrices are defined in terms of generalized supervariables in the same way  supermatrices  were defined in terms of supervariables.
The dot product for  generalized supervariables and supermatrices is defined similarly to \eqref{eq-dot-supervec} and \eq{eq-dot-supermatrix}, respectively. We also extend definitions  \eq{eq-supermatrix-tensor} and \eq{eq-supermatrix-prod}  to generalized supermatrices.
The collection   of generalized supermatrices will be denoted by  $\tilde\Ll_{m,n}$, a dense open subset of a real vector space.  Note that $\Ll_{m,n}\subset \tilde\Ll_{m,n}$.

\begin{defs}
A supersymmetry is a linear transformation $u: \Ll_{m,n}\to\tilde\Ll_{m,n}$ which leaves
$\BPhi^{\od 2}$ invariant, i.e.,
$\BPhi^{\od 2}= (u\BPhi)^{\od 2}$ {for all $\BPhi \in\Ll_{m,n}$}.
\end{defs}

Note that for $\BPhi\in\Ll_{m,n}$ the matrix $\BPhi^{\od 2}$ has entries in $\Lambda(\BPsi)$, 
which are polynomials of degree 2 in the {Grassmann generators.  If  $\BPhi\in\tilde\Ll_{m,n}$, the matrix
$\BPhi^{\od 2}$ has entries in $\Lambda(\BPsi,\xi)\supset \Lambda(\BPsi)$.}

To get some understanding of supersymmetric transformations, let us first consider 
supersymmetries that do not mix commuting and Grassmannian variables.
These supersymmetries include the orthogonal group ${\rm O}(2n)$ as follows. 
In view of \eq{eq-BPhi2}, 
given $O\in {\rm O}(2n)$ the map $\BPhi=(\bvp,\BPsi)\mapsto(\bvp O,\BPsi)$ is a supersymmetry because 
$(\bvp O)^\ot=(\bvp O)(\bvp O)^t=\bvp\bvp^t=\bvp^\ot$.
Similarly, for $S$ satisfying $S J S^t = J$ the map $(\bvp,\BPsi)\mapsto(\bvp, \BPsi S)$ is a supersymmetry, because
$(\BPsi S)^\ot=\BPsi^\ot$.
The set of such $S$ is isomorphic to the real symplectic group ${\rm Sp}(2n,\RR)$.

{Simple supersymmetric transformations   mixing commuting and Grassmannian variables are given by the maps $\Cc^p_{b,\bar b}: \Ll_{m,n}\to\tilde\Ll_{m,n}$,  with $p=1,2\ldots,n$ and $b, \bar b \, \in \RR^2$,} defined by  
\begin{equation}
\pa{\Cc^p_{b,\bar b}(\BPhi)}_{k,\ell}
: =\phi_{k,\ell}+\delta_{\ell,p} \left(2\bar b\xi \psi_{k,p}+2b\xi \overline{\psi}_{k,p} \,,\, 
 - 4\xi  \bar b\cdot \varphi_{k,p}\,,\,  4\xi  b\cdot \varphi_{k,p}\right) 
\end{equation}
for $k=1,\ldots,m$ and $\ell=1,\ldots,n$.

The dual action of supersymmetries on superfunctions is defined as follows.  
Given $F \in C^1(\RR^{m\times 2n})=C^1(\RR^{2mn})$, we extend it to a function 
$F: \Mm^{m\times n}\to \CC \oplus \xi \Lambda_\CC (\BPsi)=\Lambda_\CC(\BPsi,\xi)$ by a formal Taylor expansion,
 \beq\label{Mext}
F((x_i +\xi \psi_i)_i)= f((x_i )_i)+ \xi \sum_{j=1}^{2mn} (\partial_j F)((x_i )_i)\psi_j
\eeq  
for all $x_i \in \RR$, $\psi_i \in \Lambda^1(\BPsi)$, $i=1,2\ldots,2mn$, where by $\partial_i$ we denote the $i$-th partial derivative. Higher order terms of the Taylor expansion are not needed due to the fact that $\xi^2=0$.
Let $u:\Ll_{m,n}\mapsto\tilde\Ll_{m,n}$ be a supersymmetry. Given a supermatrix  
 $\BPhi=(\phi_{k,\ell})_{k,\ell}= (\varphi_{k,\ell},\overline{\psi}_{k,\ell},\psi_{k,\ell})_{k,\ell}$,
 we  have $u(\BPhi)_{k,\ell}= (\varphi_{k,\ell}^\prime  ,{\overline{\psi}}^\prime_{k,\ell},\psi_{k,\ell}^\prime)$ with $\varphi_{k,\ell}^\prime\in \Mm$, $\overline{\psi}_{k,\ell}^\prime,{\psi}_{k,\ell}^\prime \in \Lambda^1(\BPsi,\xi)$.  If  $\beta_{i} \in \Lambda(\BPsi)$, $i\in\{1,\ldots,2^{2mn}\}$, 
is a basis for $\Lambda(\BPsi)$, so each  $\beta_{i}$ is a polynomial in $\set{\overline{\psi}_{k,\ell},\psi_{k,\ell}}_{k,\ell}$,  we set $\beta_{i}^\prime \in  \Lambda(\BPsi,\xi)$ to be the same polynomial in the $\set{\overline{\psi}^\prime_{k,\ell},\psi^\prime_{k,\ell}}_{k,\ell}$.  Then, given $F\in C^1(\Ll_{m,n})$, we write it as in \eqref{eq-superfct-expand}, and define the function
$u^t F\in  C(\tilde\Ll_{m,n})$, where $C(\tilde\Ll_{m,n})$ is defined similarly to $C(\Ll_{m,n})$, by
\beq
u^t F(\BPhi) =  F(u\BPhi) =  \sum_i F_i ((\varphi^\prime_{k,\ell})_{k,\ell}) \beta_{i}^\prime,
\eeq
where we used \eq{Mext}.

\begin{defs}
A superfunction
$F\in  C^1(\Ll_{m,n})$ is called supersymmetric if {for all supersymmetries $u$ we have
$u^t F = F$ , i.e.,  $F(u\BPhi) = F(\BPhi)$ for all $\BPhi\in \Ll_{m,n}$.}
The set of such {supersymmetric functions will be}  denoted by $SC^1(\Ll_{m,n})$.  {We set  $SC^k(\Ll_{m,n})=SC^1(\Ll_{m,n})\cap C^k(\Ll_{m,n})$  for $k\in\NN\cup\{\infty\}$ and  
$S\Ss(\Ll_{m,n}) = SC^1(\Ll_{m,n})\cap \Ss(\Ll_{m,n})$ .}
\end{defs}

{Since supersymmetries leave $\BPhi^{\od 2}$ invariant, one may expect  that every supersymmetric function $F$ can be written 
as a function of $\BPhi^{\od 2}$, i.e., $F(\BPhi)=f(\BPhi^{\od 2})$.} This is possible in the following sense.
Let $\Sym^+(m)$ denote the non-negative, real, symmetric $m\times m$ matrices;  {clearly
$\bvp^\ot\in\Sym^+(m)$ for $\bvp  \in \RR^{m\times 2n}$.}
Let $n\geq\frac{m}{2}$, {so} the map $\bvp\mapsto \bvp^\ot=\bvp\bvp^t$ from 
$\RR^{m\times 2n}$ to $\Sym^+(m)$ is surjective {\cite[Lemma~2.5]{KSp2}}.
We denote by
$C^\infty(\Sym^+(m))$ the set of continuous  functions $f$ {on $\Sym^+(m)$ which are $C^\infty$ on the interior of $\Sym^+(m)$.  If $f\in C^\infty(\Sym^+(m))$, it follows that {$F(\bvp) = f(\bvp^\ot)$ is a continuous function on $\RR^{m\times 2n}$,  $C^\infty$ on the dense open set $\det(\bvp^\ot)\neq 0$}  \cite[Proposition~2.6]{KSp2}.}

Given a function $f\in C^\infty(\Sym^+(m))$ and $\BPhi= (\bvp,\BPsi)$ with $\det(\bvp^\ot)\neq 0$, we define  $f(\BPhi^\ot)=f(\bvp^\ot+\BPsi^\ot)$  by a formal Taylor series expansion:
\begin{equation} \label{eq-supersym-Taylor}
f(\BPhi^\ot) := f(\bvp^\ot)\,+\,\sum_{k=1}^{2mn} \frac1{k!}D^{(k)} f(\bvp^\ot)
(\BPsi^\ot,\BPsi^\ot,\ldots,\BPsi^\ot)\;,
\end{equation}
where $D^{(k)} f$ is the $k$-th derivative considered as multi-linear map $[\Sym(m)]^k\to\CC$, naturally extended to
$\Lambda_\CC(\BPsi)\otimes_\RR [\Sym(m)]^k$.
Since  $D^{(k)}f(\bvp^\ot) (\BPsi^\ot,\ldots,\BPsi^\ot) = 0$ 
for $k>2mn$, the higher order expansion terms are neglected. Given $ n \ge \frac m 2$, we define $C^\infty_n(\Sym^+(m))$ to be the subset of functions $f\in C^\infty(\Sym^+(m))$ such that there exists $F_f(\BPhi)\in C^\infty(\Ll_{m,n})$ such that $F_f(\BPhi)=f(\BPhi^\ot)$ for all  $\bvp$ with $\det(\bvp^\ot)\neq 0$. By construction  $F_f(u\BPhi) = F_f(\BPhi)$ if  $\det(\bvp^\ot)\neq 0$, so we conclude that $F_f(\BPhi)\in SC^\infty(\Ll_{m,n})$.

The following result corresponds to { \cite[Corollary~2.9]{KSp2}}.

\begin{prop} \label{prop-super-fct}{Let $n\geq\frac{m}{2}$. For all $F\in SC^\infty(\Ll_{m,n})$ there exists a unique 
$f\in C^\infty_n(\Sym^+(m))$ such that $F(\BPhi)=f(\BPhi^{\od 2})$. This establishes a bijection  from  
$SC^\infty(\Ll_{m,n})$ to $C^\infty_n(\Sym^+(m))$.}   
\end{prop}

In order to define the {appropriate function spaces we need to consider the expansion in
\eqref{eq-supersym-Taylor} in  more detail.}
We {will} reduce the general case to the case $n=1$.
{Let}  $\Phi^{(\ell)}$ denote the $\ell$-th column vector of $\BPhi$, {i.e.}
$\Phi^{(\ell)}$ is a $m\times 1$ supermatrix with entries $(\varphi_{k,\ell},\overline{\psi}_{k,\ell}, \psi_{k,\ell})_{k=1,\ldots,m}$.
Using the definition \eqref{eq-supermatrix-tensor} with $n=1$ gives
\begin{align}
((\Phi^{(\ell)})^{\od 2})_{j,k}& = 
\varphi_{j,\ell} \cdot \varphi_{k,\ell}\;+\;
\frac12\left( \overline{\psi}_{j,\ell} \,\psi_{k,\ell}\;+\;\overline{\psi}_{k,\ell}\,\psi_{j,\ell} \right)\;,\\
\label{eq-sum-replicas}
\BPhi^{\od 2} &=  \sum_{\ell=1}^n\,\left(\Phi^{(\ell)}\right)^{\od 2}\;. 
\end{align}
The different $\Phi^{(\ell)}$ are called replicas, { e.g.,}  \cite{K, KSp2}.
We also define the $m\times 2$ matrices $\varphi^{(\ell)}$ and $\Psi^{(\ell)}$, splitting
$\Phi^{(\ell)}=(\varphi^{(\ell)},\Psi^{(\ell)})$ in its real and Grassmann variables parts.

The formal Taylor expansion of $f([\Phi^{(\ell)}]^{\od 2})$ contains  only terms with monomials in $\overline{\ps_{j,\ell}},\psi_{k,\ell}$
with {equal numbers of  $\overline{\ps_{j,\ell}}$'s and  ${\ps_{j,\ell}}$'s.} 
Let $\Pp_m$ denote the set of pairs $(\bar a, a)$ of subsets of $\{1,\ldots,m\}$ with same cardinality, 
i.e.,
\begin{equation}
\Pp_m=\left\{(\bar a, a)\;: \bar a, a \subset\{1,\ldots,m\}\,,\,|\bar a|=|a|\right\}\;.
\end{equation}
For $(\bar a, a)\in\Pp_m$ and
$\bar a=\{\bar a_1,\ldots, \bar a_c\},\;a=\{a_1,\ldots,\bar a_c\}$, both ordered ({ i.e.}, $\bar a_j < \bar a_k$
and $a_j<a_k$ if $j<k$) define
\begin{equation}
\Psi^{(\ell)}_{\bar a, a} =  \prod_{k=1}^{|a|}\;\left(\,\overline{\psi}_{\bar a_k, \ell}\,\psi_{a_k,\ell}\,\right)\;,\quad
\text{with the convention}\; \Psi^{(\ell)}_{\emptyset,\emptyset} = 1
\end{equation}
For a function $f\in C^{\infty}(\Sym^+(m))$,  
let $\partial_{j,k}$ denote the derivative {with respect to} to the $j,k$-entry of the
symmetric matrix, {\it i.e.} $\partial_{j,k} f(M)=\frac{\partial}{\partial M_{k,j}}\; f(M)$. Note that $\partial_{j,k}=\partial_{k,j}$.
Furthermore, let  $\tilde \partial_{j,k} = \frac12 \partial_{j,k}$ for $j\neq k$ and  $\tilde \partial_{j,j}=\partial_{j,j}$.
Given  $(\bar a, a)\in\Pp_m$, we set $\delta_{\emptyset,\emptyset}$ to be the identity operator and
\begin{equation}
 \delta_{\bar a,a} := \det \pa{ \tilde\partial_{\bar a_r,a_s}}_{r,s=1,2,\ldots,c} =
\det\;\begin{pmatrix}
       \tilde \partial_{\bar a_1,a_1} & \cdots & \tilde \partial_{\bar a_1,a_c} &     \\
        \vdots & \ddots & \vdots   \\
	\tilde \partial_{\bar a_c,a_1} & \cdots & \tilde \partial_{\bar a_c, a_c}
      \end{pmatrix}\;  \qtx{if}  a\not= \emptyset\,.
\end{equation}
Then a Taylor expansion in the Grassmann variables yields
\beq
 f\left([\Phi^{(\ell)}]^{\od 2}\right) = 
\sum_{(\bar a, a)\in\Pp_m}\;
\delta_{\bar a,a} f\left([\varphi^{(\ell)}]^\ot \right)\;
\Psi^{(\ell)}_{\bar a, a}\; ,
\eeq
which for $n=1$ {is the same as} \eqref{eq-supersym-Taylor}.
For the general case we use \eqref{eq-sum-replicas} and an {iterated} Taylor expansion.
Thus, for $(\bar\aaa, \aaa) = (\bar\aaa_\ell, \aaa_\ell)_{\ell=1,\ldots,n} \;\in\;(\Pp_m)^n=\Pp_m^n$ we define
\begin{equation} \label{eq-def-repl-prod}
 \Psi_{\bar\aaa,\aaa} := 
\prod_{\ell=1}^n \Psi^{(\ell)}_{\bar\aaa_\ell, \aaa_\ell}\;,\quad
D_{\bar \aaa, \aaa} :=  \prod_{\ell=1}^n \delta_{\bar\aaa_\ell, \aaa_\ell}\; ,
\end{equation}
{getting}
\begin{equation}\label{eq-super-Taylor}
 f(\BPhi^{\od 2}) = 
\sum_{ (\bar\aaa, \aaa)\,\in\,\Pp_m^n}
D_{\bar\aaa,\aaa}\; f(\bvp^\ot)\;\Psi_{\bar\aaa,\aaa}\;.
\end{equation}

From this formula one can obtain an interesting Leibniz-type formula.
Let $(\bar\aaa,\aaa), (\bar\bbb,\bbb) \in \Pp_m^n$. If $\bar\aaa_\ell\cap \bar\bbb_\ell= \aaa_\ell\cap \bbb_\ell=\emptyset $ for each $\ell=1,\ldots,n$,  we define
$(\bar\aaa+\bar\bbb, \aaa+\bbb) \in\Pp_m^n$ by
$(\bar\aaa+\bar\bbb)_\ell = \bar\aaa_\ell \cup \bar\bbb_\ell$ and $(\aaa+\bbb)_\ell=\aaa_\ell\cup\bbb_\ell$.
In this case we also define $\sgn(\bar\aaa,\aaa,\bar\bbb,\bbb)\in\{-1,1\}$ by
\beq
\Psi_{\bar\aaa,\aaa}\,\Psi_{\bar\bbb,\bbb} = 
\sgn(\bar\aaa,\aaa,\bar\bbb,\bbb)\; \Psi_{\bar\aaa+\bar\bbb,\aaa+\bbb}\;.
\eeq
Since the product of two smooth supersymmetric functions is smooth and supersymmetric, we obtain for all
$f, g\in {C^\infty_n(\Sym^+(m))}$ and all $(\bar\aaa,\aaa)\in\Pp_m^n$ that
\beq \label{eq-Leibn}
D_{\bar\aaa,\aaa}\,(fg) = 
\sum_{\bar\bbb+\bar\bbb'=\bar\aaa\;,
\bbb+\bbb'=\aaa}\;\sgn(\bar\aaa,\aaa,\bar\bbb,\bbb)\,
D_{\bar\bbb,\bbb}\, g\,D_{\bar\bbb',\bbb'}\, f\;.
\eeq

%%%%%%%%%%%%%%%%%%%%%%%%%%%%%%%%%%%%%%%%%%%%%%%%%%%%%%%%%%%%%%%%%%%%%%%%%%%%%%%%%%%%%%%%
%%%%%%%%%%%%%%%%%%%%%%%%%%%%%%%%%%%%%%%%%%%%%%%%%%%%%%%%%%%%%%%%%%%%%%%%%%%%%%%%%%%%%%%%
%%%%%%%%%%%%%%%%%%%%%%%%%%%%%%%%%%%%%%%%%%%%%%%%%%%%%%%%%%%%%%%%%%%%%%%%%%%%%%%%%%%%%%%%
%%%%%%%%%%%%%%%%%%%%%%%%%%%%%%%%%%%%%%%%%%%%%%%%%%%%%%%%%%%%%%%%%%%%%%%%%%%%%%%%%%%%%%%%

\subsection{The supersymmetric Fourier transform and  Banach spaces of supersymmetric functions}

 From now on we  fix $n\geq\frac{m}{2}$ and, in view of Proposition~\ref{prop-super-fct},  we make the following identifications: 
\beq  \begin{split} 
&SC^{\infty}(\Ll_{m,n}) \cong C^\infty_n(\Sym^+(m)) := \set{f\in C^{\infty}(\Sym^+(m)): \; F(\BPhi)=f(\BPhi^{\od 2})\in SC^{\infty}(\Ll_{m,n}) }, \\
& S\Ss(\Ll_{m,n})\cong{ \Ss_n(\Sym^+(m))}   : =\left\{f\in {C^\infty_n(\Sym^+(m))}\;:\; F(\BPhi):=f(\BPhi^{\od 2}) \in \Ss(\Ll_{m,n}) \right\}. 
\end{split}  \eeq

{The supersymmetric Fourier transform $T$ will {play an important role in our analysis.}   Given} $f\in S\Ss(\Ll_{m,n})$ we define $Tf\in S\Ss(\Ll_{m,n})$ by
\begin{equation}\label{eq-def-T}
 (Tf)((\BPhi')^{\od 2}) = 
\int e^{\imath \BPhi'\cdot\BPhi}\, f({\BPhi}^{\od 2})\; D\BPhi\; ,
\end{equation}
where we use the fact   that the right hand side defines a  supersymmetric function 
\cite{KSp2}. 
The integral {with respect to} $D\BPhi$ only sees terms multiplied
by $\Psi_{\ccc,\ccc}$ where $(\ccc,\ccc)\in\Pp_m^n$ with $\ccc_\ell=\{1,\ldots,m\}$ for all $\ell$,
{\it i.e.} all sets in $(\ccc,\ccc)$ are the complete set $\{1,\ldots,m\}$.
In other words $\Psi_{\ccc,\ccc}=\prod_{k,\ell}\overline{\psi}_{k,\ell}\psi_{k,\ell}$, so
 $\int {\Psi}_{\ccc,\ccc} D\Psi_{\ccc,\ccc} = (-1)^{mn}$, where
$D\Psi_{\ccc,\ccc} = \prod_{k,\ell} d\overline{\psi}_{k,\ell}\,d\psi_{k,\ell}$.
For $(\bar\aaa,\aaa)=(\bar\aaa_\ell,\aaa_\ell)_\ell \in \Pp_m^n$ we define the complement 
$(C\bar\aaa, C\aaa)\in\Pp_m^n$ by 
{ $\bar\aaa+C\bar\aaa=\ccc$ and $\aaa+C\aaa=\ccc$.
Setting $\sgn(\bar\aaa,\aaa)=(-1)^{mn} \sgn(\bar\aaa,\aaa,C\bar\aaa,C\aaa)$, we have}
\begin{eqnarray}
 \Psi_{\bar\aaa,\aaa}\;\Psi_{C\bar\aaa,C\aaa}&=&
(-1)^{mn}\; \sgn(\bar \aaa, \aaa)\;\Psi_{\ccc,\ccc}\; ,\\
\sgn(\bar\aaa,\aaa) &=&
\int \Psi_{\bar\aaa,\aaa}\;\Psi_{C\bar\aaa,C\aaa}\,D\Psi_{\ccc,\ccc}\;.
\end{eqnarray}
Clearly  $\sgn(\bar\aaa,\aaa)=\sgn(C\bar\aaa,C\aaa)$. Interchanging each $ \overline{\psi}_{k,\ell}$ with $ {\psi_{k,l}}$ in
 $\Psi_{\bar\aaa,\aaa}\;\Psi_{C\bar\aaa,C\aaa}$ gives a sign of $(-1)^{mn}$ because {one applies a permutation} 
consisting of $mn$ transpositions. {One then has a product of the form  $\prod_{x,y} \psi_{x}\overline{\psi}_{y}$.}
Changing {this product} to $\prod_{x,y} \overline{\psi}_{y}\psi_{x}$ gives another sign of $(-1)^{mn}$. After these two changes one has 
switched $\bar\aaa$ and $\aaa$. Therefore
$\Psi_{\bar\aaa,\aaa}\Psi_{C\bar\aaa,C\aaa} = \Psi_{\aaa,\bar\aaa}\Psi_{C\aaa, C\bar\aaa}$ and hence the signs are the same.
To summarize we get
\begin{equation}\label{eq-sgn-relations}
\sgn(\bar\aaa,\aaa) = \sgn(C\bar\aaa,C \aaa) = 
\sgn(\aaa,\bar\aaa) = \sgn(C\aaa,C\bar\aaa).
\end{equation}

In order to relate the components in the Grassmann variables we need to expand
$e^{\imath \BPhi'\cdot \BPhi}$, but the only terms that matter for the supersymmetric  integral are with
Grassmann monomials of the form $\Psi_{\bar\aaa,\aaa}$.
As
\begin{equation}
e^{i \BPhi'\cdot\BPhi} = 
e^{i \Tr(\bvp' \bvp^t)}\;\prod_{k,\ell} \left[\left(1+\tfrac{i}{2} \overline{\psi'}_{k,\ell}\psi_{k,\ell}\right)
\left(1+\tfrac{i}{2} \overline{\psi}_{k,\ell}\psi'_{k,\ell}\right)\right]\; ,
\end{equation}
 the factor $\Psi_{\bar\aaa,\aaa}$ in the expansion of $e^{-i\Tr(\varphi'\varphi^t) }e^{i \BPhi'\cdot \BPhi}$ 
appears as
\begin{align} 
\left(\tfrac{i}{2}\right)^{2|\aaa|} \prod_{\ell=1}^n \prod_{k=1}^{|\aaa_\ell| }
\overline{\psi}_{\bar a_\ell^k, \ell} \psi'_{\bar a_\ell^k, \ell} \overline{\psi'}_{a_\ell^k,\ell}\psi_{a_\ell^k,\ell}&  = 
\left(\tfrac{-1}{4}\right)^{|\aaa|} (-1)^{|\aaa|}\prod_{\ell=1}^n \prod_{k=1}^{|\aaa_\ell|}
\overline{\psi'}_{a_\ell^k, \ell} \psi'_{\bar a_\ell^k,\ell} \overline{\psi}_{\bar a_\ell^k, \ell} \psi_{a_\ell^k,\ell}\\
&= 
\tfrac{1}{4^{|\aaa|}}\,\Psi'_{\aaa,\bar\aaa}\; \Psi_{\bar\aaa,\aaa}\; , \nonumber
\end{align}
where $|\aaa|=\sum_{\ell=1}^n |\aaa_\ell|$ and
$\bar \aaa_\ell=\{\bar a_\ell^1, \ldots ,\bar a_\ell^{|\bar \aaa_\ell|}\},\;
\aaa_\ell = \{a_\ell^1,\ldots, a_\ell^{|\aaa_\ell|} \}$.
Thus, given  $f\in S\Ss(\Ll_{m,n})$,  we get
\begin{align}
(Tf)(\BPhi'^{\od 2}) & = 
\int e^{i \Tr(\bvp'\bvp^t)}
\left[ \sum_{(\bar\aaa,\aaa)\in\Pp_m^n} \tfrac{1}{4^{|\aaa|}}\,\Psi'_{\aaa,\bar\aaa}\; \Psi_{\bar\aaa,\aaa}\right]
\sum_{(\bar \bbb,\bbb)\in\Pp_m^n} D_{\bar\bbb,\bbb} f(\bvp^\ot) \Psi_{\bar\bbb,\bbb}\;
D\BPhi \\
& = \frac 1 {\pi^{mn}}\sum_{(\bar \aaa, \aaa)\in\Pp_m^n} \tfrac{1}{4^{|\aaa|}} \Psi'_{\aaa,\bar\aaa} \;\sgn(\bar\aaa,\aaa)
\int_{\RR^{m\times 2n}} \!\!\!\!\!\! e^{i \Tr(\bvp'\bvp^t)}\;
D_{C\bar\aaa, C\aaa}\; f(\bvp^\ot)\;d^{2mn}\bvp.\notag
\end{align}
In particular,  one has
\begin{equation}\label{eq-T-parts}
D_{\bar\aaa,\aaa} (Tf) = 
\tfrac{2^{mn}}{4^{|\aaa|}}\;\sgn(\aaa,\bar\aaa)\;\Ff(D_{C\aaa,C\bar\aaa}\,f) \qtx{for all} (\bar\aaa,\aaa) \in \Pp_m^n.
\end{equation}
Here  $\F$ denotes the Fourier transform on $\RR^{m\times 2n}$; we abuse the notation by letting $\F f$ denote  
the function in $\Ss_n(\Sym^+(m))$ such that $(\F f)(\bvp^\ot)$ is the Fourier transform of the function  $F(\bvp)=f(\bvp^\ot)$.
Using $|\aaa|+|C\aaa| = mn$, $\sgn(\aaa,\bar\aaa) = \sgn(C\bar \aaa,C\aaa)$, and the fact that
the inverse Fourier transform $\Ff^*$ and $\Ff$ {coincide} on functions invariant under $\bvp\mapsto-\bvp$,  
{we conclude that for all
$f\in\Ss(\Ll_{m,n})$ we have
\beq
D_{\bar\aaa,\aaa}\; (TTf) = 
D_{\bar\aaa,\aaa}f  \qtx{for all} (\bar\aaa,\aaa)\,\in\,\Pp_m^n.
\eeq
Thus,
\beq
T^2f =  TTf = f\; \qtx{for all} f\in{ \Ss_n(\Sym^+(m))}.
\eeq}

Following Campanino and Klein \cite{CK,K,KSp2}, we introduce  the norms $\hn  \cdot  \hn_p$ on 
 ${ \Ss_n(\Sym^+(m))}\cong S\Ss(\Ll_{m,n})$, $p\in [1,\infty)$, given by
\begin{equation} \label{eq-def-norms}
\hn f \hn_p^2\;: =\;
\sum_{(\bar\aaa,\aaa)\in\Pp_m^n}\; \left\|\,2^{|\aaa|}\;
D_{\bar\aaa,\aaa} \;f\,(\bvp^\ot)\right\|^2_{L^p(\RR^{m\times 2n}, d^{2mn}\bvp)}.
\end{equation}
We  define the Hilbert space $\widehat{\Hh}$ as completion of ${ \Ss_n(\Sym^+(m))}$ with respect to the norm $\hn\cdot\hn_2$.  The Banach spaces  $\widehat{\Hh}_p$,  $p \in [1,\infty]$, are defined by
\begin{equation}\label{eq-def-spaces}
\widehat{\Hh}_p\; : =\;\{f\;\in\;\widehat{\Hh}\;:\; \|f\|_{\widehat{\Hh}_p}\; :=\;\hn f \hn_2 +\hn f \hn_p\;<\; \infty\;\}\;.
\end{equation}

The supersymmetric Fourier transform $T$, defined in \eq{eq-def-T} as an operator on ${ \Ss_n(\Sym^+(m))}$, extends to $\widehat\Hh$ as a unitary operator in view of  \eqref{eq-T-parts}-\eqref{eq-def-norms}.  Moreover, as $\abs{e^{\imath\Tr(\bvp'\bvp^t)}}\le 1$, we get
$\hn Tf \hn_\infty\leq (2\pi)^{-mn}\,\hn f \hn_1$, {so}
$T$ is a bounded operator from $\widehat\Hh_1$ to $\widehat\Hh_\infty$. 
\begin{remark}  The spaces $\widehat{\Hh}$, $\widehat{\Hh}_p$, the supersymmetric Fourier transform $T$, etc., all depend on our choice of $n\ge \frac m 2$ for a given $m$.  This dependence on $n$ (and $m$) will be generally omitted.
\end{remark}

For technical reasons we will work on closed subspaces ${\Hh}$ and ${\Hh}_p$ of $\widehat{\Hh}$ and $\widehat{\Hh}_p$.   
For a given  a complex symmetric $m\times m$ matrix $B$ with strictly positive real part
(i.e., $\re B > 0$), let $\PE(B)$ denote the vector space spanned by functions $f \in C^\infty_n(\Sym^+(m))$ of the form
$f(M) = p(M) \exp(-\Tr(MB))$, where $p(M)$ is a polynomial in the entries of $M\in\Sym^+(m)$.  As derivatives only introduce polynomial factors, we have
$\PE(B) \subset{ \Ss_n(\Sym^+(m))}$. We now define $\PE(m) \subset{ \Ss_n(\Sym^+(m))}$ as the smallest vector space containing $\PE(B)$ for all complex symmetric $m\times m$ matrices $B$ with strictly positive real part.
($\PE$ stands for ``polynomial times exponential''.)
  We define $\Hh$ and  $\Hh_p$ as the closures of $\PE(m)$ in $\widehat\Hh$ and  
$\widehat\Hh_p$, respectively.

\begin{lemma} \label{lem-PEB} Given any    complex symmetric $m\times m$ matrix $B$  with $\re B > 0$ and $p\in [1,\infty)$,
 $\Hh$ and  $\Hh_p$ are the closures of $\PE(B)$ in $\widehat\Hh$ and  
$\widehat\Hh_p$, respectively.
\end{lemma}

\begin{proof} Let  $\Xi$ denote the collection of complex symmetric $m\times m$ matrices $B$  with $\re B > 0$.   Note that $\Xi$ can be identified with an  
open connected subset of $\CC^{\frac 1 2 m(m+1)}$. Given $B\in\Xi$, we set  $\alpha_{B} = \min \sigma (\re B)>0$.   
For each $B\in \Xi$ and $p\in [1,\infty)$ we let  $\Hh^{(B)}$ and  $\Hh_p^{(B)}$ denote the closures of $\PE(B)$ in $\widehat\Hh$ and  
$\widehat\Hh_p$, respectively.

Given $B\in \Xi$, let $C$ be a complex symmetric $m\times m$ matrix  such that $\norm{C} \le \frac 1 2\alpha_{B}$, so $B+C \in \Xi$ with $\alpha_{B+C}\ge \frac 1 2\alpha_{B}$. Consider
\begin{align}
 h_n(\bvp^\ot)   =   p(\bvp^\ot)\,\left\{e^{-\Tr(B \bvp^\ot)}\left[ \sum_{s=0}^n \frac1{s!} (-\Tr(C\bvp^\ot))^s\right]\;-\;e^{-\Tr((B+C)\bvp^\ot)}\right\}\;, 
\end{align}
where $p(M)$ is a polynomial on the entries of $M \in \Sym^+(m)$.  We have uniform bounds in $n$:
\beq
\abs{h_n(\bvp^\ot)}\le 2\,|p(\bvp^\ot)|e^{-\Tr((\re B)\bvp^\ot)+|\Tr(C\bvp^\ot)|}\le 2 |p(\bvp^\ot)|e^{-\frac 1 2\alpha_{B} \Tr(\bvp^\ot)}.
\eeq
 Letting $n\to\infty$, $h_n$ converges point-wise to zero.
Similar statements hold for the derivatives $D_{\bar\aaa,\aaa} h_n$. 
Thus, $h_n\to 0$ in   $\Hh_p$   by dominated convergence.  It follows that $\PE(B+C)\subset   \Hh_p^{(B)}$, and hence $\Hh_p^{(B+C)}\subset \Hh_p^{(B)}$.

If  $B\in \Xi$ and  $C$ is a complex symmetric $m\times m$ matrix with  $\norm{C} \le \frac 1 3\alpha_{B}$, it follows that $\alpha_{B+C}\ge \frac 2 3 \alpha$, so $\norm{C} \le \frac 1 2\alpha_{B+C}$, and we have  $\Hh_p^{(B)}\subset \Hh_p^{(B+C)}$.  We thus conclude that $\Hh_p^{(B+C)}= \Hh_p^{(B)}$ if $\norm{C} \le \frac 1 3\alpha_{B}$.

Thus, for all $B \in \Xi$ we have that  $\Xi_B:=\set{B^\prime \in \Xi : \Hh_p^{(B^\prime)}= \Hh_p^{(B)}}$ 
and $\Xi\setminus \Xi_B$ are open subsets of $\Xi$.
Since $\Xi$ is connected we conclude that for all $B \in \Xi$ we
have $\Xi_B=\Xi$, and hence  $\Hh_p^{(B)}=\Hh_p$ for all $p\in[1,\infty)$.

The same argument applies to $\Hh$, which is the same as $\Hh_{2}$ except for a scalar factor in the norm.
\end{proof}

The supersymmetric Fourier transform $T$ maps Gaussian functions, 
%i.e., functions of the type  
$\bvp \mapsto \exp(-\Tr(B\bvp^\ot))$, into Gaussian functions. 
Letting {$B=B_0+tB_1$} and expanding in $t$, one recognizes that $T$ leaves $\PE(m)$ invariant. 
It follows that $T$ is a unitary operator on $\Hh$ and  a bounded operator from $\Hh_1$ to $\Hh_\infty$.
This proves  the following lemma.

\begin{lemma} \label{lem-T} 
The supersymmetric Fourier transform $T$ is a unitary operator on $\Hh$ and $\widehat\Hh$, and  a bounded operator from $\Hh_1$ to $\Hh_\infty$ and from $\widehat\Hh_1$ to $\widehat\Hh_\infty$. 
\end{lemma}

%%%%%%%%%%%%%%%%%%%%%%%%%%%%%%%%%%%%%%%%%%%%%%%%%%%%%%%%%%%%%%%%%%%%%%%%%%%%%%%%%%%%%%%%
%%%%%%%%%%%%%%%%%%%%%%%%%%%%%%%%%%%%%%%%%%%%%%%%%%%%%%%%%%%%%%%%%%%%%%%%%%%%%%%%%%%%%%%%
%%%%%%%%%%%%%%%%%%%%%%%%%%%%%%%%%%%%%%%%%%%%%%%%%%%%%%%%%%%%%%%%%%%%%%%%%%%%%%%%%%%%%%%%
%%%%%%%%%%%%%%%%%%%%%%%%%%%%%%%%%%%%%%%%%%%%%%%%%%%%%%%%%%%%%%%%%%%%%%%%%%%%%%%%%%%%%%%%

\section{The averaged matrix Green's  function} \label{sec-avgreen}

%%%%%%%%%%%%%%%%%%%%%%%%%%%%%%%%%%%%%%%%%%%%%%%%%%%%%%%%%%%%%%%%%%%%%%%%%%%%%%%%%%%%%%%%
%%%%%%%%%%%%%%%%%%%%%%%%%%%%%%%%%%%%%%%%%%%%%%%%%%%%%%%%%%%%%%%%%%%%%%%%%%%%%%%%%%%%%%%%
%%%%%%%%%%%%%%%%%%%%%%%%%%%%%%%%%%%%%%%%%%%%%%%%%%%%%%%%%%%%%%%%%%%%%%%%%%%%%%%%%%%%%%%%
%%%%%%%%%%%%%%%%%%%%%%%%%%%%%%%%%%%%%%%%%%%%%%%%%%%%%%%%%%%%%%%%%%%%%%%%%%%%%%%%%%%%%%%%

We fix an arbitrary site in $\B$ which we will call the origin and denote by $0$.  
Given two nearest neighbors 
sites $x,y \in \B$, we denote by ${\B}^{(x|y)}$ the lattice  obtained from $\B$ by removing
the branch emanating from $x$ that passes through $y$; if we do not specify which 
branch was removed we will simply write ${\B}^{(x)}$.  Each vertex in  ${\B}^{(x)}$ has degree 
$K +1$, with the single exception of $x$ which has degree $K$. 
Given $\La \subset \B$, we will use $\,H_{\lb, \La}$ to denote the operator $\,H_{\lb}$  
restricted to $\ell^2 (\Lambda,\CC^m)$ with Dirichlet boundary condition  {(i.e., free boundary condition)}.
The matrix Green's function corresponding to $\,H_{\lb, \La}$ will be denoted by  	
\begin{equation}
G_{\lb, \La}\, (x,y;z)\;\; :=\;\;\left[\left\langle {x,j|(H_{\lb, \La} -z)^{-1}|y,k} \right\rangle\right]_{j,k\in\{1,\ldots,m\}},
\end{equation}
where $x,y \in {\La}$ 
and $z = E + i\eta$ with $E \in \RR$, $\eta> 0$.
Special important choices of $\Lambda$ are the sets $\B_L$, denoting all sites {$y \in \B$ with  $d(0,y)\leq L$}, and
$\B^{(x|y)}_L$ denoting all sites $x' \in \B^{(x|y)}$ with $d(x,x')\leq L$.
We will use the Green's matrix at the origin very often, therefore let us define
\beq
G_\lb(z)\;{:=}\; G_\lb(0,0;z)\;.
\eeq

For special choices of $\Lambda$ let us also introduce the following notations:
\beq
\begin{array}{lclclcl}
H_{\lb,L}            &{:=}&  H_{\lb,\B_L}      &\qquad& 
G_{\lb,L}(z)         &{:=}&  G_{\lambda,\B_L}(0,0;z)  \\
H_{\lb}^{(x|y)}      &{:=}&  H_{\lb, {\B}^{(x|y)}} & &
G_{\lb}^{(x|y)}(z)   &{:=}&  G_{\lb,{\B}^{(x|y)}}(x,x;z) \\
H_{\lb,L}^{(x|y)}    &{:=}&  H_{\lb, {\B}^{(x|y)}_L} & &
G_{\lb,L}^{(x|y)}(z) &{:=}&  G_{\lb,{\B}^{(x|y)}_L}(x,x;z) \\
H_{\lb}^{(x)}        &{:=}&  H_{\lb, {\B}^{(x)}} & &
G_{\lb}^{(x)}(z)     &{:=}&  G_{\lb,{\B}^{(x)}}(x,x;z) \\
\end{array}
\eeq

To each site $x\in\B$ we assign independent supermatrix variables $\BPhi_x = 
(\bvp_x,\BPsi_x)\,\in\,\Ll_{m,n}(\BPsi_x)$, i.e. $\bvp_x$ is a variable varying in
$\RR^{m\times 2n}$ and $\BPsi_x=((\overline{\psi}_x)_{k,\ell},(\psi_x)_{k,\ell})_{k,\ell}$
where the $(\psi_x)_{k,\ell}, (\overline{\psi}_x)_{k,\ell}$ are all independent
Grassmann variables. 
Let $B$ be an operator on $\ell^2(\B,\CC^m)$ and $B_\Lambda$ its restriction to $\ell^2(\Lambda,\CC^m)$ for a subset
$\Lambda\subset\B$. 
For $x,y\in\Lambda$ we define $\langle x|B_\Lambda|y \rangle$ 
to be the $m\times m$ matrix with entries $(\langle x,j|B_\Lambda|y,k\rangle)_{j,k}$.
Furthermore, for a finite subset $\Lambda\subset\B$, we
define 
\beq D_\Lambda\BPhi = \prod\limits_{x\in\Lambda} D\BPhi_x\;,
\eeq 
with $D\BPhi_x$ as in \eqref{eq-def-DPhi}, and
\beq
\langle \BPhi|B_\Lambda-z|\BPhi\rangle =  \sum_{x,y\in\Lambda}\, \BPhi_x\,\cdot\;\langle x\,|\,B_\Lambda-z\,|\,y\,\rangle\,\BPhi_y\;.
\eeq
We will use this notation for $B_\Lambda =H_{\lb,L}$ and $B_\Lambda=H^{(x|y)}_{\lb,L}$, where $\Lambda=\B_L$ or $\B_L^{(x|y)}$.

We now take  $\im z >0$, $\Lambda\subset\B$  finite, and $x,y\in\Lambda$, and state 
 identities that are crucial for our analysis.
 
  If $B_\Lambda$ is an operator on $\ell^2(\Lambda,\CC^m)$, 
 symmetric (i.e., $\langle x,j|B_\Lambda|y,k\rangle=\langle y,k |B_\Lambda|x,j\rangle$) with a
strictly positive real part, then
 \cite[Theorem~III.1.1]{KSp2}
\begin{equation}\label{eq-SSint}
 \int e^{-\langle \BPhi|B_\Lambda|\BPhi\rangle}\, D_{\Lambda} \BPhi = 1\;.
\end{equation}

The supersymmetric replica trick \cite{B,E,K} gives  
\begin{equation}\label{eq-SSrpltrick}
 \left[G_{\lambda,\Lambda}(x,y;z)\right]_{j,k} = 
i\int (\psi_x)_{j,\ell} (\overline{\psi}_y)_{k,\ell}\,e^{-i\langle\BPhi|H_{\lambda,\Lambda}-z|\BPhi\rangle}\,D_{\Lambda}\BPhi 
\end{equation} 
for any  $\ell\in\{1,\ldots,n\}$.

As $H_{\lambda,\Lambda}$ is a real operator ({ i.e.,} it leaves $\ell^2(\Lambda,\RR^m)$ invariant) the resolvent
$G_{\lambda,\Lambda}$ is symmetric. This means
$
\left[G_{\lambda,\Lambda}(x,y;z)\right]_{j,k}  = \left[G_{\lambda,\Lambda}(y,x;z)\right]_{k,j} 
$.
Therefore one can replace $(\psi_x)_{j,s} (\overline{\psi}_y)_{k,s}$ in \eq{eq-SSrpltrick}  by
\beq
-\frac1n \left[\BPsi_x J \BPsi^t_y\right]_{j,k} = - \frac{1}{n} \sum_{\ell=1}^n \frac12 \left[ (\overline{\psi}_y)_{k,\ell}(\psi_x)_{j,\ell}
+\overline{\psi}_x)_{j,\ell}(\psi_y)_{k,\ell} \right],
\eeq
obtaining
\begin{equation} \label{eq-Gr}
 G_{\lambda,\Lambda}(x,y;z) = 
-\frac{i}{n} \int \BPsi_x J \BPsi^t_y\,e^{-i\langle\BPhi|H_{\lambda,\Lambda}-z|\BPhi\rangle}\,D_{\Lambda}\BPhi\,.
\end{equation}

 In particular,
\begin{align} \label{eq-Greensmatrix-1}
& G_{\lambda,L}(0,0;z) = -\frac{i}{n}
\int  \BPsi_0^\ot \, e^{-i\,\langle 
\BPhi| H_{\lambda,L}\,-z |\BPhi\rangle}\,D_{\B_L}\BPhi\;\\ \notag
& \quad =
 -\frac{i}{n} \int  \BPsi_0^\ot \, e^{i\BPhi_0\cdot \,(z-\lb V(0)-A)\,\BPhi_0} \Bigg[ \prod_{\substack{x\in \B\\d(x,0)=1}} e^{-i \BPhi_0\,\cdot\,\BPhi_x-i\langle \BPhi|H^{(x|0)}_{\lambda,L-1}\,-z|\BPhi\rangle}
D_{\B^{(x|0)}_{L-1}}\BPhi \Bigg]  D\BPhi_0 . \notag
\end{align}
In order to simplify this equation one uses
\begin{equation} \label{eq-Greensmatrix-2}
\int\, e^{-i\BPhi_0\cdot\BPhi_x- i\langle \BPhi|H^{(x|0)}_{\lambda,L-1} -z|\BPhi\rangle} 
\,D_{\B^{(x|0)}_{L-1}}\BPhi
 = 
e^{(i/4)\,\BPhi_0 \cdot G^{(x|0)}_{\lambda,L-1}(z)\BPhi_0}\;,
\end{equation}
which can be obtained from \eqref{eq-SSint} by completing the square.
{We plug \eqref{eq-Greensmatrix-2} into \eqref{eq-Greensmatrix-1},  take the limit $L$ to infinity,
and write $\BPhi$ for the supermatrix variable $\BPhi_0$, obtaining}
\begin{equation}
\label{eq-Greensmatrix-3}
G_\lb(z) = 
-\frac{i}{n}\int {\BPsi^\ot} \; e^{i\BPhi \cdot (z-\lb V(0)-A)\BPhi}\:
e^{\left(\frac{i}{4}\! \sum\limits_{x:d(x,0)=1} \BPhi \cdot G^{(x|0)}_{\lambda} \BPhi \right)}\;D\BPhi.
\end{equation}

If in \eqref{eq-Greensmatrix-2} we repeat the argument used in \eqref{eq-Greensmatrix-1} and let $L\to\infty$,  then
 for $d(x,0)=1$ we get
\begin{equation}\label{eq-Greensmatrix-4}
 e^{\frac i 4\,\BPhi\,\cdot\, G^{(x|0)}_{\lambda}(z)\BPhi} = 
\int e^{-i\BPhi \cdot \BPhi'} e^{i  \BPhi'\,\cdot\,(z-\lambda V(x)-A)\BPhi'}
e^{\left(\frac{i}{4}\!\! \sum\limits_{y:d(y,x)=1, y\neq 0}\!\! \BPhi' \,\cdot\, G^{(y|x)}_{\lambda}(z) \BPhi' \right)} D\BPhi'.
\end{equation}

For any $z=E+i\eta$ in the upper half plane, i.e., $\eta>0$, we define  $\ze_{\lb,z}\in  C^\infty_n(\Sym^+(m))$ by
\begin{equation}\label{eq-zeta}
\ze_{\lb,z} (\bvp^{\od 2}) =  \E\left( e^{ \frac{i}{4} \Tr( G_{\lb}^{(0)}(z)\,\bvp^{\od 2})}\right)
 = \E\left(\e^{\frac{i}{4}\,\bvp\,\cdot\,G_{\lb}^{(0)}(z)\,\bvp} \right)\;.
\end{equation}

\begin{theorem}   For any $\lb \in \RR$, $E \in \RR$ and $\eta > 0$ one has
\begin{equation}
\E(G_{\lb} (z)) = 
 -\frac{i}{n} \int   {\BPsi^\ot}\,e^{i\BPhi\cdot(z -A)\BPhi}\, h(\lb  \BPhi^{\od 2}) \,[\ze_{\lb,z} (\BPhi^{\od 2})]^{K+1}\,  D\BPhi  \;,
\label{eq-EG}
\end{equation}
and 
\begin{equation} \label{eq-zeta-recursion}
\ze_{\lb,z} (\BPhi^{\od 2}) = \int e^ { - i \BPhi\cdot\BPhi'}\,
 \left\{e^{i \BPhi'\,\cdot\,(z-A)\BPhi' }\, h(\lb  \BPhi'^{\od 2}) \,[\ze_{\lb,z} (\BPhi'^{\od 2})]^K \right\}\, D\BPhi' \;.
\end{equation}
\end{theorem}

\begin{proof}  If we  take expectations in  \eqref{eq-Greensmatrix-3} and  \eqref{eq-Greensmatrix-4}, with respect to the potential's 
probability distribution, and recall that the $V(x)$,  $x \in {{\B}}$,  are independent, identically 
distributed random variables, we get \eqref{eq-EG} and \eqref{eq-zeta-recursion}.
\end{proof}

 Note that the Hamiltonian $H_0$ (i.e., $\lambda=0$) splits into a direct sum of shifted Laplacians on $m$ copies of the
Bethe lattice. The Laplacians are   shifted by the energies $a_i$, $i=1,\ldots,m$, where $A=\diag(a_1,\ldots,a_m)$.
In this case we can calculate $G_{0}^{(0)} (z) $  as in \ci{AK}
and obtain
\begin{equation}
\ze_{0,z} (\bvp^{\od 2}) =  \prod_{k=1}^m 
e^{ \frac{ i }{2K } \left[\,-z+a_k +\sqrt{(z-a_k)^2 - K}\,\right]\varphi_k^2 } \;,\label{eq-ze0}
\end{equation}
where $\varphi_k^2=\sum_\ell \varphi_{k,\ell}\cdot \varphi_{k,\ell}$ 
and \mbox{Im $\sqrt{ \ } >0$}.  If $E\,\in\, I_{A,K}$, {\it i.e.} 
$ |E-a_k| <\sqrt{ K}$ for all $k=1,\ldots,m$, then we have the point-wise limit  
\begin{equation}
\ze_{0,E} (\bvp^{\od 2})\;{:=}\;\lim_{\eta \downarrow 0} \ze_{0,z} (\bvp^{\od 2}) =
e^{-i \bvp\cdot A_E\,\bvp}, \label{eq-ze00}
\end{equation}
where $A_E$ is the diagonal matrix 
\begin{align}\label{eq-def-AE}
A_E &= \tfrac{1}{2K }\left((E-A)  \,-\,i\sqrt{ K- (E-A)^2} \right)\;, \qtx{i.e.,}\\
 (A_E)_{k,k^\prime}&= \tfrac{1}{2K }\pa{(E-a_k)  -i\sqrt{ K- (E-a_k)^2} }\delta_{k,k^\prime}, \quad k,k^\prime =1,2,\ldots,m. \notag
\end{align}

In order to write \eqref{eq-zeta-recursion} in a compact way, let us introduce
the operator 
\beq
B_{\lb,z} =M(e^{i\bvp \cdot (z-A) \bvp } h(\lb  \bvp ^{\od 2})),
\eeq
where for a given function $g\in C^\infty_n(\Sym^+(m))$ we use $M(g)$, or $M(g(\bvp^\ot))$, 
to denote the operator given by multiplication by $g(\bvp^\ot)$:   
\begin{equation}
(M(g)f)(\bvp^\ot) \,:=\, g(\bvp^\ot)f(\bvp^\ot)\;.
\end{equation}
Then  \eqref{eq-zeta-recursion} can be written as $\zeta_{\lb,z} = TB_{\lambda,z} \zeta_{\lb,z}^K$ with $T$ as defined in \eqref{eq-def-T}.
The crucial observation is that this is a fixed point equation in $\Hh_\infty$.

\begin{prop} \label{zeta} We have:
\begin{enumerate}
\item[{\rm (i)}] For $\eta=\im z\geq 0$ the operator $B_{\lb,z}$ is a bounded operator on $\widehat\Hh_1$, leaving  $\Hh_1$ invariant, and the map $(\lambda,E,\eta,f)\mapsto TB_{\lambda,E+i\eta} f^K$ defines a continuous map from 
$\RR\times\RR\times[0,\infty)\times \Hh_\infty$ to $\Hh_\infty$.

\item[{\rm (ii)}]   $\ze_{\lb,z} \in {\Hh}_\infty$ for all $\lb \in \RR$ and $z=E+i\eta$ with 
$\eta >0$.   The map $(\lb,E, \eta) \to \ze_{\lb,E +i\eta}$ is continuous from
$\RR \times \RR \times (0,\infty)$ to $ {\Hh}_\infty$.

\item[{\rm (iii)}]  If $E\,\in\,I_{A,K}$, then  $\ze_{0,E} \in {\Hh}_\infty$ and
\begin{equation}
\lim_{\eta \downarrow 0} \ze_{0,E+i\eta} = \ze_{0,E} \;\;\;\mbox{in}\;\; {\Hh}_\infty  \;.   
\end{equation}

\item[{\rm (iv)}] {The equality}  \eqref{eq-zeta-recursion} can be rewritten as a fixed point equation in 
${\Hh}_\infty$: 
\begin{equation} \label{eq-fpz}
\ze_{\lb,z}  = TB_{\lb,z} \ze_{\lb,z}^K \;,
\end{equation}
valid for all $\lb \in \RR$ and $z=E+i\eta$ with $\eta >0$, and also valid for  $\lb=0$ and
 $z=E$ with $E\,\in\, I_{A,K}$. 
\end{enumerate}
\end{prop}

\begin{proof}
(i): There are polynomials $p_{\bar\aaa,\aaa}$ defined on the complex symmetric matrices such that
for any complex symmetric matrix $B$ one has
\beq
e^{i \BPhi\cdot B\BPhi} = \sum_{(\bar\aaa,\aaa)\in\Pp_m^n}\,
p_{\bar\aaa,\aaa}(B) e^{i \Tr(B \bvp^\ot)}\;\Psi_{\bar\aaa,\aaa}\;.
\eeq
Therefore
\beq
e^{i\BPhi\cdot (z-A) \BPhi} h(\lb  \BPhi^{\od 2}) =
\sum_{(\bar\aaa,\aaa)\in\Pp_m^n} \E\left[ p_{\bar\aaa,\aaa}(z-A-\lambda V(0)) e^{i \Tr(z-A-\lb V(0)) \bvp^\ot)} \right] \Psi_{\bar\aaa,\aaa} ,
\eeq
since all moments of the random variable $V(0)$ are finite.
It follows that
\begin{equation}
\hn B_{\lb,z} f \hn_p \;\leq\;   {C_{m,n,p} \,}\left[ \max_{(\bar\aaa,\aaa)}
\E  \left| p_{\bar\aaa,\aaa}(z-A-\lb V(0) \right|\right] \hn f \hn_p  \qtx{for} p \in [1,\infty). 
\end{equation}
In particular,  $B_{\lb,z}$  is a bounded operator on  $\widehat\Hh_1$. In order to show that it leaves $\Hh_1$ invariant it suffices to show that $\PE(m)$ is mapped to $\Hh_1$.

Let $f(\bvp^\ot)=p(\bvp^\ot)e^{-Tr(B\bvp^\ot)}\,\in\,\PE(m)$. 
If the distribution $\mu$ of $V(0)$ is a point measure then it is easy to see that $B_{\lambda,z} f\,\in\, \PE(m)\subset\Hh_1$.
If the distribution $\mu$ has compact support, we take a sequence of point probability measures $\mu_n$, $\supp \mu_n \subset \supp \mu$,
 which converge weakly to $\mu$.  Let $B^n_{\lambda,z}$ denote the corresponding operators 
when $\mu$ is replaced by $\mu_n$. As $p_{\bar\aaa,\aaa}(z-A-\lb V(0))$ can be replaced by bounded continuous functions in $\lb V(0)$ 
(deviating from the original function outside the support of $\mu$) 
we get by weak convergence that $B_{\lambda,z}^n f \in\Hh_1$ converges point-wise  on $\Sym^+(m)$, together with  all its derivatives, to $B_{\lb,z} f$. By dominated convergence this is true with respect to  the $\hn\cdot\hn_1$ and $\hn\cdot\hn_2$ norms, hence
$B_{\lb,z}f \in \Hh_1$.
Finally if $\mu$ is a measure such that all moments exist, we approximate it by the 
compactly supported measures $\mu 1_{\{\|V\|<n\}}$ and use dominated convergence again to obtain $B_{\lambda,z} f \in \Hh_1$.

Let $\lambda_n \to \lambda$, $z_n \to z$ with non-negative imaginary parts, and $f_n\to f$ in $\Hh_1$.  Since
$\|B_{\lambda_n,z_n} f_n - B_{\lambda,z,} f\|_{\Hh_1}\le \|B_{\lambda_n,z_n} f_n-B_{\lambda_n,z_n} f\|_{\Hh_1}+\|B_{\lambda_n,z_n} f - B_{\lambda,z} f\|_{\Hh_1}$
it follows from the uniform boundedness of $B_{\lambda_n,z_n}$ in operator norm and dominated convergence 
that both converge to zero.
Hence $(\lambda,E,\eta,f)\mapsto B_{\lambda,E+i\eta}f$ is a continuous map from $\RR\times\RR\times[0,\infty)\times \Hh_1$
to $\Hh_1$.

Let $f, g_n \in\widehat\Hh_\infty$, $\norm{f -g_{n}}_{\widehat{\Hh}_{\infty}} \to 0$.
Using \eqref{eq-Leibn} and H\"older's inequality in various ways
one sees that
$\hn f^K-g_n^K\hn_p=\hn (f-g_n)(f^{K-1}+f^{K-2}g+\ldots+g^K)\hn_p$ converges to zero for $p=1,2$.
Moreover the map $f\mapsto f^K$ leaves $\PE(m)$ invariant. Hence $f\mapsto f^K$ defines a continuous map from
$\widehat \Hh_\infty$ to $\widehat\Hh_1$, mapping $\Hh_\infty$ to $\Hh_1$.
Continuity of  $(\lambda,E,\eta,f)\mapsto T B_{\lb,E+i\eta}f^K$ 
from $\RR\times\RR\times[0,\infty)\times \Hh_\infty$ to $\Hh_\infty$ now follows from
the continuity of $(\lambda,E,\eta,f)\mapsto B_{\lb,E+i\eta} f$ as shown above  since 
$T$ is continuous from $\Hh_1$ to $\Hh_\infty$.

 (ii): If $\eta >0$, we can similarly show that $B_{\lb,z}\ze_{\lb,z}^K  \in \Hh_p$ for all $p\in [1,\infty]$.
It follows from  \eqref{eq-zeta-recursion} that  \eqref{eq-fpz} holds, so  
$\ze_{\lb,z} \in {\Hh}_\infty$  for all $\lb \in \RR$ and 
\mbox{$z=E+i\eta$} with $\eta >0$.  To prove the continuity, note that for any fixed potential $\Vv$,
\mbox{$G^{(0)}_{\lb} ( E +i\eta)$}  is a continuous function of  
$  (\lb, E,\eta) \in \RR \times \RR \times (0, \infty)$
with {$  \norm{G^{(0)}_{\lb}\, ( E +i\eta)} \le \frac{1}{\eta} $}.  The continuity
in (ii) then follows from the Dominated Convergence Theorem.

 Part (iii) is proven by explicit computations and dominated convergence.  
Finally, part (iv)  follows from \eqref{eq-zeta-recursion} and parts (i)-(iii). In particular, 
\eqref{eq-fpz} for  $\lb=0$ and $z=E$ with $E\in I_{A,K}$ follows from continuity arguments.
 \end{proof}

%%%%%%%%%%%%%%%%%%%%%%%%%%%%%%%%%%%%%%%%%%%%%%%%%%%%%%%%%%%%%%%%%%%%%%%%%%%%%%%%%%%%%%%%
%%%%%%%%%%%%%%%%%%%%%%%%%%%%%%%%%%%%%%%%%%%%%%%%%%%%%%%%%%%%%%%%%%%%%%%%%%%%%%%%%%%%%%%%
%%%%%%%%%%%%%%%%%%%%%%%%%%%%%%%%%%%%%%%%%%%%%%%%%%%%%%%%%%%%%%%%%%%%%%%%%%%%%%%%%%%%%%%%
%%%%%%%%%%%%%%%%%%%%%%%%%%%%%%%%%%%%%%%%%%%%%%%%%%%%%%%%%%%%%%%%%%%%%%%%%%%%%%%%%%%%%%%%

\section{The averaged squared matrix Green's  function} \label{sec-avsqgreen}

%%%%%%%%%%%%%%%%%%%%%%%%%%%%%%%%%%%%%%%%%%%%%%%%%%%%%%%%%%%%%%%%%%%%%%%%%%%%%%%%%%%
%%%%%%%%%%%%%%%%%%%%%%%%%%%%%%%% Start equations with \xi  %%%%%%%%%%%%%%%%%%%%%%%%
%%%%%%%%%%%%%%%%%%%%%%%%%%%%%%%%%%%%%%%%%%%%%%%%%%%%%%%%%%%%%%%%%%%%%%%%%%%%%%%%%%%
%%%%%%%%%%%%%%%%%%%%%%%%%%%%%%%%%%%%%%%%%%%%%%%%%%%%%%%%%%%%%%%%%%%%%%%%%%%%%%%%%%

In order to obtain absolutely continuous spectrum, we consider the expectation of $|G_\lb(z)|^2 = (G_\lb(z))^* G_\lb(z)$.
To do so, let us introduce  independent supermatrices $\BPhi_+=(\bvp_+,\BPsi_+)$ and $\BPhi_-=(\bvp_-,\BPsi_-)$.
Furthermore, let $C^\infty((\Sym^+(m))^2)$ denote the {set of continuous functions on $(\Sym^+(m))^2$ that  are
$C^\infty$ on its interior.} For $f\in C^\infty((\Sym^+(m))^2)$, 
$\det \bvp_+^\ot \neq 0$ and $\det\bvp_-^\ot \neq 0$, define
\beq \label{eq-2expand}
f(\BPhi_+^\ot,\BPhi_-^\ot)=
\sum_{(\bar\aaa,\aaa),(\bar\bbb,\bbb) \in \Pp_m^n} D^{(+)}_{\bar\aaa,\aaa} D^{(-)}_{\bar\bbb,\bbb} f( \bvp_+^\ot,\bvp_-^\ot)\,
(\BPsi_+)_{\bar\aaa,\aaa}(\BPsi_-)_{\bar\bbb,\bbb}\,,
\eeq
where $D^{(+)}_{\bar\aaa,\aaa}$ denotes the differential operator
$D_{\bar\aaa,\aaa}$ with respect to the entries of the matrix $\bvp_+^\ot$ and
$D^{(-)}_{\bar\bbb,\bbb}$ the operator
$D_{\bar\bbb,\bbb}$ with respect to  entries of the matrix $\bvp_-^\ot$.
We define $C^\infty_n((\Sym^+(m))^2)$ to be the set of functions
$f\in C^\infty((\Sym^+(m))^2)$ {such that  $F_f(\BPhi_+,\BPhi_-)=f(\BPhi_+^\ot,\BPhi_-^\ot)$ extends
to a $C^\infty$ superfunction.}  We also define the subspace $\Ss_n((\Sym^+(m))^2)$ of functions
$C^\infty_n((\Sym^+(m))^2)$ {such that} $F_f$ is in the Schwartz space.

{Note that by construction $F_f(\BPhi_+,\BPhi_-)$ is {separately} supersymmetric in both {supervariables}, i.e.,  for supersymmetries
$u_+, u_-$ acting on $\Ll_{m,n}(\BPsi_+)$ and $\Ll_{m,n}(\BPsi_-)$ respectively, one has
$F_f(u_+\BPhi_+, u_-\BPhi_-)=F_f(\BPhi_+,\BPhi_-)$.
}

For $\lb \in \RR$, $E \in \RR$ and $\eta > 0$ let us introduce $\xi_{\lb,z}\in C^\infty_n((\Sym^+(m))^2)$
\begin{equation}
 \xi_{\lb,z} ( \bvp_+^\ot\,,\, \bvp_-^\ot) \,=\,
\E \left(\exp{ \left\{\frac{i}{4} \Tr\left ( G_{\lb}^{(0)}(z) \,\bvp_+^\ot\, - \,
\overline{ G_{\lb}^{(0)}(z) }\,\bvp_-^\ot\right)\right\}}\right)    \label{xia} \;.
\end{equation}
Note that $\overline{ G_{\lb}^{(0)}(z) } = \left[ G_{\lb}^{(0)}(z) \right]^*$ since it is a symmetric matrix.

\begin{theorem}
One has
\begin{align}\label{eq-EGG}
\E\,\left| G_{\lb}(z)\right|^2 &= \tfrac{1}{n^2}\int\;
 \BPsi_+^\ot\, \BPsi_-^\ot\;
e^{i [\BPhi_+ \cdot(z-A)\BPhi_+\,-\,\BPhi_-\cdot(\bar z -A)\BPhi_-]}\;
\\ 
&  \; \qquad \qquad \qquad \times\; h(\lambda(\BPhi_+^{\od 2}-\BPhi_-^{\od 2}))\;
\left[\xi_{\lb,z}(\BPhi_+^{\od 2} , \BPhi_-^{\od 2})\right]^{K+1}\;D\BPhi_+\,D\BPhi_-  \notag
\end{align}
and
\begin{align} \label{eq-xi-recursion}
 \xi_{\lb,z} ( \BPhi_+^\ot, \BPhi_- ^\ot) & =  
\int\;
e^ { - i (\BPhi_+\cdot\BPhi_+' -  \BPhi_-\cdot\BPhi_-')}\;
e^{i [\BPhi_+ \cdot(z-A)\BPhi_+\,-\,\BPhi_-\cdot(\bar z -A)\BPhi_-]}\;\\
  & \qquad \qquad \quad \times\; h(\lb ({\BPhi_+'}{^\ot}-{\BPhi_-'}{^\ot}))
[\xi_{\lb,z} ( {\BPhi_+'}{^\ot}, {\BPhi_-'}{^\ot})]^K \; D\BPhi_+'\, D\BPhi_-'  \nonumber  \;.
\end{align}
\end{theorem}

\begin{proof}  From \eqref{eq-Greensmatrix-3} we get
%@GG
\begin{align}\label{GG}
\left| G_{\lb} (z)\right|^2&  = 
\tfrac{1}{n^2}\int \BPsi_+^\ot\, \BPsi_-^\ot\;
e^{i \BPhi_+\cdot(z-\lb V(0)-A)\BPhi_+ - i \BPhi_-\cdot(\bar z - \lb V(0)-A)\BPhi_-}\; \\
&  \qquad   \times \exp\left(\frac{i}{4}\,\sum\limits_{x:d(x,0)=1} (\BPhi_+\cdot G_\lb^{(x|0)}(z) \BPhi_+\,-\,
\BPhi_- \cdot \overline{G_\lb^{(x|0)}(z)} \BPhi_-)\right)\, D\BPhi_+\,D\BPhi_-  . \notag
\end{align}
Taking expectations we get \eqref{eq-EGG}.  To prove \eqref{eq-xi-recursion}, we use \eqref{eq-Greensmatrix-4}, \eqref{xia}, 
and take expectations.
\end{proof}   

For $\lb=0$ we have {
\begin{equation}
 \xi_{0,z} ( \bvp_+^\ot, \bvp_- ^\ot)\, = \, \ze_{0,z}( \bvp_+^\ot) \overline{ \ze_{0,z}( \bvp_-^\ot)} \;.
  \label{xi0}
\end{equation}}
Again, as in (\ref{eq-ze00}), when $ E\,\in\,I_{A,K}$ we have the point wise limit {
\begin{equation}
 \xi_{0,E} ( \bvp_+^\ot, \bvp_- ^\ot)\,=\,\lim_{\eta \downarrow 0}\xi_{0,z} ( \bvp_+^\ot, \bvp_- ^\ot) = 
 e^{-i\bvp_+\,\cdot\, A_E\,\bvp_+\;+\;i\bvp_-\,\cdot\,\overline{A_E}\,\bvp_-}  \;,\label{xi00}
\end{equation}}
with $A_E$ as defined in \eqref{eq-def-AE}.

We want to {rewrite} \eqref{eq-xi-recursion} as a fixed point equation {similar to
\eqref{eq-fpz}. To do so} we need to introduce some tensor spaces and tensor norms.
First let us introduce $\PE(m)^{\otimes 2} = \PE(m)\otimes \PE(m)$ which is the vector space spanned by functions
$g(\bvp_+^\ot,\bvp_-^\ot) =  f(\bvp_+^\ot) \tilde f(\bvp_-^\ot)$ where
$f,\tilde f\,\in\,\PE(m)$.
On $\PE(m)^{\otimes 2}$ we define for $1\le p \le \infty$ the tensor norms
\begin{equation}
\hnn g \hnn_p^2  \; = 
\sum_{\substack{(\bar\aaa,\aaa)\in \Pp_m^n\\(\bar\bbb,\bbb)\in \Pp_m^n}}
\left\| 2^{|\aaa|+|\bbb|}\,D^{(+)}_{\bar \aaa,\aaa}\, D^{(-)}_{\bar\bbb,\bbb} \, g(\bvp_+^\ot, \bvp_-^\ot) \,
\right\|^2_{{L}^p(\bvp_+,\bvp_-)} ,
\end{equation}
where $\|\cdot\|_{L^p(\bvp_+,\bvp_-)}$ denotes the $p$-norm of the $L^p$ space on ${\left(\RR^{m\times 2n}\right)}^2$
in the variables $\bvp_+,\bvp_-$ {with respect to} the Lebesgue measure $d^{2mn} \bvp_+\, d^{2mn} \bvp_-$.
Now the Hilbert space tensor product $\K=\Hh\otimes\Hh$ is the completion of $\PE(m)^{\otimes 2}$ {with respect to}
$\hnn\cdot\hnn_2$.
Furthermore we set ${\T} = T \otimes T$, so $\T$ is unitary on $\K$.  We also define {
\begin{equation}
{\cB}_{\lb,z} = 
M(\e^{i [\bvp_+\cdot(z-A)\bvp_+\,-\,\bvp_-\cdot(\bar z -A)\bvp_-]} h(\lb (\bvp_+^\ot-\bvp_- ^\ot))  )\;,
\end{equation}}
where
$M(g(\bvp_+^\ot,\bvp_-^\ot ))$ denotes multiplication by the function $g(\bvp_+^\ot,\bvp_-^\ot ) $.

To handle the nonlinear equations  \eqref{eq-EG}, \eqref{eq-zeta-recursion},  \eqref{eq-EGG}) and \eqref{eq-xi-recursion},  we  
introduce the Banach spaces
\begin{equation}
{\K}_p   =  \{g \in {\K}\;,\;\;\|g\|_{{\K}_p}  =  \hnn g\hnn_2 +\hnn g\hnn_p  < \infty \}  \;,
\end{equation}
for $1 \le p \le \infty$.

%@xita
\begin{prop}  \label{xita} We have:
\begin{enumerate}
\item[{\rm (i)}] For $\eta=\im z\geq 0$ the operator $\cB_{\lb,z}$ is a bounded operator on $\K_1$.
For $K\geq 2$ the map $(\lambda,E,\eta,g)\mapsto \T\cB_{\lambda,E+i\eta} g^K$ defines a continuous map
from $\RR\times\RR\times [0,\infty)\times \K_\infty$ to $\K_\infty$.

\item[{\rm (ii)}] $\xi_{\lb,z} \in {\K}_\infty$ for all $\lb \in \RR$ and $z=E+i\eta$ with 
$\eta >0$.   The map $(\lb,E, \eta) \to \xi_{\lb,E +i\eta}$ is continuous from
$\RR \times \RR \times (0,\infty)$ to $ {\K}_\infty$.

\item[{\rm (iii)}]  If $E\,\in\,I_{A,K}$, then
 $\xi_{0,E} \in {\K}_\infty$ and
\begin{equation}
\lim_{\eta \downarrow 0} \xi_{0,E+i\eta} = \xi_{0,E} \;\;\;\mbox{in}\;\; {\K}_\infty   \;. 
\end{equation}

\item[{\rm (iv)}]  {The equality} \eqref{eq-xi-recursion} can be rewritten as a fixed point equation in 
${\K}_\infty$:   %@eq-fpx 
\begin{equation}\label{eq-fpx}
\xi_{\lb,z}   = {\T}{\cB}_{\lb,z} \xi_{\lb,z}^K \;,
\end{equation}
valid for all $\lb \in \RR$ and $z=E+i\eta$ with $\eta >0$, and also valid for  $\lb=0$ and
 $z=E$ with  $E\in I_{A,K}$. 
 \end{enumerate}
\end{prop}

\begin{proof}  The proof is completely analogous to the proof of Proposition~\ref{zeta}. 
\end{proof}

%%%%%%%%%%%%%%%%%%%%%%%%%%%%%%%%%%%%%%%%%%%%%%%%%%%%%%%%%%%%%%%%%%%%%%%%%%%%%%%%%%%%%%%%
%%%%%%%%%%%%%%%%%%%%%%%%%%%%%%%%%%%%%%%%%%%%%%%%%%%%%%%%%%%%%%%%%%%%%%%%%%%%%%%%%%%%%%%%
%%%%%%%%%%%%%%%%%%%%%%%%%%%%%%%%%%%%%%%%%%%%%%%%%%%%%%%%%%%%%%%%%%%%%%%%%%%%%%%%%%%%%%%%
%%%%%%%%%%%%%%%%%%%%%%%%%%%%%%%%%%%%%%%%%%%%%%%%%%%%%%%%%%%%%%%%%%%%%%%%%%%%%%%%%%%%%%%%

\section{A fixed point analysis} \label{sec-fxp}

%%%%%%%%%%%%%%%%%%%%%%%%%%%%%%%%%%%%%%%%%%%%%%%%%%%%%%%%%%%%%%%%%%%%%%%%%%%%%%%%%%%%%%%%
%%%%%%%%%%%%%%%%%%%%%%%%%%%%%%%%%%%%%%%%%%%%%%%%%%%%%%%%%%%%%%%%%%%%%%%%%%%%%%%%%%%%%%%%
%%%%%%%%%%%%%%%%%%%%%%%%%%%%%%%%%%%%%%%%%%%%%%%%%%%%%%%%%%%%%%%%%%%%%%%%%%%%%%%%%%%%%%%%
%%%%%%%%%%%%%%%%%%%%%%%%%%%%%%%%%%%%%%%%%%%%%%%%%%%%%%%%%%%%%%%%%%%%%%%%%%%%%%%%%%%%%%%%

In this section we will analyze the fixed point equations \eqref{eq-fpz} and \eqref{eq-fpx} in more detail.
A crucial ingredient is given by the following lemma.

We let $\Delta(m,\ZZ_+)$ denote the collection of  $m\times m$  upper triangular matrices with non-negative integer entries. Given   $J=(J_{j,k})_{j,k}\in \Delta(m,\ZZ_+)$, we set $|J| = \sum_{j\leq k} J_{j,k}$.

\begin{lemma} \label{lemma-CE}
Let $E\,\in\,I_{A,K}$ and define the operator $C_E = TB_{0,E} M(\zeta_{0,E}^{K-1})$. 
\begin{enumerate}
\item[{\rm (i)}] $C_E$ is a bounded operator on $\Hh$ and $\Hh_\infty$. Moreover,
$C_E^2$ is a compact operator on  $\Hh$ and $\Hh_\infty$. 

\item[{\rm (ii)}]
The eigenvalues of    $C_E$ as an operator on the Hilbert space  $\Hh$ are given by
\begin{equation} \label{eq-lbj}
 \lambda_{J} = \prod_{\substack{j,k\in\{1,\ldots,m\}\\j\leq k}} \left[4\;(A_E)_{j,j} (A_E)_{k,k}\right]^{J_{j,k}}\;,
\quad  J=(J_{j,k})_{j,k}\in \Delta(m,\ZZ_+),
\end{equation}
where 
$A_E$ is the diagonal matrix  defined in \eqref{eq-def-AE}.\\
The corresponding eigenfunctions $\set{f_{J} : J\in \Delta(m,\ZZ_+)}$  are of the form
\begin{equation}\label{eq-fj}
 f_J(\bvp^\ot) = \left((\bvp^\ot)^J\,+\,p_J(\bvp^\ot)\right)\,\zeta_{0,E}(\bvp^\ot)\in\PE(m)\;,
\end{equation}
where
 \beq
(\bvp^\ot)^J\; : =\;\prod\limits_{\substack{j,k\in\{1,\ldots,m\}\\j\leq k}} \left[(\bvp^\ot)_{j,k}\right]^{J_{j,k}}
\eeq
is a monomial of degree $|J|$ and $p_J$ is a polynomial of degree strictly less than $|J|$.   
Moreover,  
\beq  \label{eq-lbj2}
\lambda_J\not= K^{-1}  \qtx{and} |\lambda_J| = {K^{- |J|}} \qtx{for all}  J\in \Delta(m,\ZZ_+),
\eeq 
and
\beq \label{eq-spectrum-CE}
K^{-1} \notin \sigma_\Hh(C_E)= \{\lambda_J: \, J\in \Delta(m,\ZZ_+) \}\cup\{0\}.
\eeq

\item[{\rm (iii)}] The spectrum of  $C_E$ as an operator on $\Hh_\infty$ is the same
as  its spectrum as an operator on $\Hh$:  
\beq \label{eq-spectrum-CEinfty}
\sigma_{\Hh_\infty}(C_E) = \sigma_\Hh(C_E).
\eeq
\end{enumerate}
\end{lemma}

\begin{proof}
(i):  Since $\zeta_{0,E} \in\PE(m)$, $M(\zeta_{0,E}^{K-1})$ is a bounded operator on $\widehat \Hh$ leaving
$\Hh$ invariant. Using H\"older's inequality one also realizes that $M(\zeta_{0,E}^{K-1})$ is
a bounded operator from $\widehat \Hh_\infty$ to $\widehat \Hh_1$, 
mapping $\Hh_\infty$ to $\Hh_1$.
Hence by Lemma~\ref{lem-T} and Proposition~\ref{zeta}~(i) the operator $C_E = TB_{0,E}M(\zeta_{0,E}^{K-1})$ 
is  a bounded operator on $\widehat \Hh$ and $\widehat \Hh_\infty$, leaving $\Hh$ and $\Hh_\infty$ invariant.
Compactness of $C_E^2$ on $\Hh$ and $\Hh_\infty$ will follow from compactness of $C_E^2$ on $\widehat{\Hh}$ and $\widehat{\Hh}_\infty$.
The proof is now completely analogous to \cite[Proposition~III.1.6]{KSp2}.
As shown in \cite{AK,KSp2},  if  $\beta_1, \beta_2\in{ \Ss_n(\Sym^+(m))}$ are compactly supported smooth functions, the operator
$M(\beta_1(\bvp^\ot)) \Ff M(\beta_2(\bvp^\ot))$ is compact on $L^2(\RR^{m\times 2n})$
and on $C_b(\RR^{m\times 2n})$, the bounded
 continuous functions.
This, 
combined with \eqref{eq-T-parts} and the Leibniz rule \eqref{eq-Leibn}, implies that
$M(\beta) T M(\beta)$ is compact as an operator on $\widehat \Hh$ and also as an operator from $\widehat \Hh_\infty$ 
to $\widehat \Hh_1$ for $\beta\in{ \Ss_n(\Sym^+(m))}$ with compact support.

From \eqref{eq-ze00} we see that $\ze_{0,E}\in\PE(m)$, hence all its derivatives are exponentially decaying functions 
of $\bvp^\ot$. Therefore, using dominated convergence,  one can approximate 
$D_E{:=} B_{0,E } M( \ze_{0,E}^{K-1})TB_{0,E } M( \ze_{0,E}^{K-1})$  
in operator norm, both as an operator on  $\widehat\Hh$ and as an operator from  $\widehat\Hh_\infty$ to   $\widehat\Hh_1$, by operators 
{$M(\beta)T M(\beta)$ with smooth, compactly supported $\beta$}. 
Hence $D_E$ is compact on $\widehat{\Hh}$ and from $\widehat{\Hh}_\infty$ to $\widehat \Hh$.
Therefore $C_E^2=TD_E$ is compact on $\widehat{\Hh}$ and $\widehat{\Hh}_\infty$.

To obtain (ii) let us start with the identity   
\begin{align}
& C_E\; e^{it \Tr(M\BPhi^\ot)} \zeta_{0,E}(\BPhi^\ot) = TB_{0,E} \,e^{it \Tr(M\BPhi^\ot)} \zeta_{0,E}^K(\BPhi^\ot)  \label{eq-CE} \\
& \qquad \qquad  = \int e^{-i\BPhi'\cdot\BPhi} e^{i\BPhi'\cdot(E-A-KA_E+tM) \BPhi'}\;D\BPhi' 
= e^{\frac{i}{4} \Tr((A-E+KA_E-tM)^{-1} \BPhi^ \ot)}\; ,\nonumber
\end{align}
where $M$ is a real symmetric matrix and $t\in \RR$. \eqref{eq-CE} is derived from \eqref{eq-SSint} by completing the square.

Let $\Pp_s(\BPhi^\ot)$ denote the set of homogeneous
polynomials of degree $s$ in the entries of $\BPhi^\ot$, together with the zero polynomial to make it a vector-space.
{Furthermore, let $\Pp_{\leq s}(\BPhi^\ot)$ and $\Pp_{<s}(\BPhi^\ot)$ denote the polynomials in the entries of $\BPhi^\ot$ of degree
smaller or equal to $s$ and strictly less than $s$, respectively.}

Using the identity $(K A_E+A-E)^{-1} = -4A_E$ as well as \eqref{eq-ze00},
a Taylor expansion {with respect to} $t$ of the right hand side of \eqref{eq-CE} gives
\begin{align}\notag
& e^{\frac{i}{4} \Tr\left((A-E+KA_E-tM)^{-1} \BPhi^\ot\right)} = 
e^{\frac i4 \Tr(-4A_E\BPhi^\ot)}
e^{\frac{i}{4} \sum_{s=1}^\infty t^s \Tr\left((-4A_E M)^s (-4A_E) \BPhi^\ot\right)} \\
 & \qquad \qquad \qquad =
\zeta_{0,E}(\BPhi^\ot)\left[1+\sum_{s=1}^\infty \frac{(it)^s}{s!}\left(\left[\Tr(4 A_E M A_E \BPhi^\ot)\right]^s 
+p_{s,M}(\BPhi^\ot) \right) \right],
\end{align}
where $p_{s,M}\in\Pp_{<s}(\BPhi^\ot)$.
Performing a Taylor expansion of the left hand side of
\eqref{eq-CE} and comparing terms leads to
\begin{equation}\label{eq-CE-2}
 C_E\,\left(\Tr(M \BPhi^\ot)\right)^s  \zeta_{0,E}(\BPhi^\ot)  = 
\left[\left(\Tr(4A_E M A_E \BPhi^\ot)\right)^s \,+\,p_{s,M}(\BPhi^\ot) \right]\zeta_{0,E}(\BPhi^\ot)\;.
\end{equation}
By linearity of $C_E$, the map
$[\Tr(M\BPhi^\ot)]^s \mapsto [\Tr(4A_E M A_E \BPhi^\ot)]^s+p_{s,M}(\BPhi^\ot)$, varying $M$, 
can be extended to a linear map from $\Pp_s(\BPhi^\ot)$ to $\Pp_{\leq s}(\BPhi^\ot)$.
Since the natural projection from $\Pp_{\leq s}(\BPhi^\ot)$ onto $\Pp_{s}(\BPhi^\ot)$ is linear,
the map $[\Tr(M\BPhi^\ot)]^s \mapsto [\Tr(4A_E M A_E \BPhi^\ot)]^s$ can be extended to a linear map 
$\gamma_s:\Pp_s(\BPhi^\ot)\to\Pp_s(\BPhi^\ot)$.
Using all real symmetric matrices $M$, the polynomials of the form $[\Tr(M\BPhi^\ot)]^s$ 
span $\Pp_s(\BPhi^\ot)$. Hence the extension is unique.

Expanding these homogeneous polynomials one obtains 
\beq \label{eq-expand1}
[\Tr(M\BPhi^\ot)]^s=\sum_{\substack{j_1,\ldots,j_s\\k_1,\ldots,k_s}} \prod_{i=1}^s M_{j_i,k_i} 
(\BPhi^\ot)_{j_i,k_i}
 = \sum_{\substack{J\in\Delta(m,\ZZ_+)\\|J|=s}} \!\!c(M,J) (\BPhi^\ot)^J\,,
\eeq
where the latter equation defines the coefficients $c(M,J)$.
Similarly, since $A_E$ is diagonal,
\begin{align} \label{eq-expand2}
[\Tr(4 A_E M A_E \BPhi^\ot)]^s &=\sum_{\substack{j_1,\ldots,j_s\\k_1,\ldots,k_s}} \prod_{i=1}^s 
M_{j_i,k_i}\,4(A_E)_{j_i,j_i} (A_E)_{k_i,k_i} (\BPhi^\ot)_{j_i,k_i} \nonumber \\
&=\sum_{\substack{J\in\Delta(m,\ZZ_+)\\|J|=s}} \lambda_J \, c(M,J) (\BPhi^\ot)^J\,.
\end{align}
Thus, we conclude that
\beq
\gamma_s\pa{(\BPhi^\ot)^J} = \lambda_J (\BPhi^\ot)^J \qtx{for all} J\in\Delta(m,\ZZ_+) \;\:\text{with}\; \; |J|=s.
\eeq
Therefore, \eqref{eq-CE-2} implies
\begin{equation}
 C_E\, (\BPhi^\ot)^J\,\zeta_{0,E}(\BPhi^\ot)  =  \left[\lambda_J (\BPhi^\ot)^J\;+\; \tilde p_J(\BPhi^\ot)\right]\,
\zeta_{0,E}(\BPhi^\ot),
\end{equation}
where $\tilde p_J \in \Pp_{<|J|}(\BPhi^\ot)$. 

\eq{eq-lbj2} follows   from \eq{eq-lbj} and \eqref{eq-def-AE} by explicit computations.  As $|\lambda_J|=K^{-|J|}$,  one has
$\lambda_J \neq \lambda_{J'}$ whenever $|J|>|J'|$.  Performing an induction with respect to $|J|$ 
yields eigenfunctions of the form \eqref{eq-fj} for the eigenvalues $\lambda_J$.

The linear span of the eigenfunctions $f_J$,  $J\in \Delta(m,\ZZ_+)$, is $\PE(-iA_E)$. It follows from Lemma~\ref{lem-PEB} that their closed linear span is $\Hh$, so we get \eq{eq-spectrum-CE}.

 For part (iii) note that $\sigma_{\Hh_\infty}(C_E)\subset \sigma_\Hh(C_E)$ by compactness of
$C_E^2$ in $\Hh_\infty$. Equality follows as all eigenfunctions $f_J$ are in $\PE(m)\subset \Hh_\infty$.
\end{proof}

\begin{remark} To obtain \eqref{eq-spectrum-CE} we work with the space $\Hh_\infty$,
{\it i.e.} the closure of $\PE(m)$ {with respect to} the norm $\hn f\hn_2+\hn f\hn_\infty$ rather than the closure of the Schwartz functions,
$\widehat{\Hh}_\infty$.
\end{remark}

\begin{lemma}  \label{F}  The map
 $F:\,\RR \times \RR  \times [0,\infty)  \times  {\Hh}_\infty \to {\Hh}_\infty$,  defined by
\begin{equation}
F(\lb, E,\eta,f) =  TB_{\lb,E + i \eta} f^K -f  \;,
\end{equation}
 is continuous.  The map $F$ is continuously Frechet differentiable with respect to $f$, the partial 
derivative being
\begin{equation}
F_f(\lb, E,\eta,f) = K  TB_{\lb,E + i \eta} M( f^{K-1}) - I \;.
\end{equation}
Moreover, for any $E\,\in\,I_{A,K}$ we have $ F(0, E,0,\ze_{0,E}) =0$ and %@00
\begin{equation}
0 \notin \si( F_f(0, E,0,\ze_{0,E}))  \;.\label{00}
\end{equation}     
\end{lemma}

\begin{proof}   Continuity  follows from Proposition~\ref{zeta}, differentiability is straightforward.
So let us show \eqref{00}.  We have
\beq F_f(0, E,0,\ze_{0,E}) = K C_{E}-I .
\eeq
 Recalling \eq{eq-spectrum-CEinfty} and  \eq{eq-spectrum-CE}, we get
 \beq \label{eq-spectrum-F}
 \si( F_f(0, E,0,\ze_{0,E})))= \{K\lambda_J-1: \, J\in \Delta(m,\ZZ_+) \}\cup\{-1\}.
\eeq
Thus, \eqref{00} follows immediately from \eq{eq-lbj}.
\end{proof}

\bl  \label{Q}  The map
 $Q:\,\RR \times \RR  \times [0,\infty)  \times  {\K}_\infty \to {\K}_\infty$,  defined by
\begin{equation}
Q(\lb, E,\eta,g) = {\T} {\cB}_{\lb,E + i \eta} g^K -g  \;,
\end{equation}
 is continuous.  $Q$ is continuously Frechet differentiable with respect to $g$, the partial 
derivative being
\begin{equation}
Q_g(\lb, E,\eta,g) = K {\T} {\cB}_{\lb,E + i \eta} M( g^{K-1}) - I \;.
\end{equation}
Moreover, for any $E\,\in\,I_{A,K}$ we have $ Q(0, E,0,\xi_{0,E}) =0$ and %@000
\begin{equation}
0 \notin \si( Q_g(0, E,0,\xi_{0,E})) \;. \label{000}
\end{equation}     
\el

\begin{proof}  The first two statements are completely analogous to the previous lemma.
We have to show  \eqref{000}. 
  $ Q_g(0, E,0,\xi_{0,E}) = K{\Cc}_{E} - I$ where $\Cc_{E}=  {\T}{\cB}_{0,E } M( \xi_{0,E}^{K-1}) $ . 
 It follows from (\ref{xi0}) that  ${\Cc}_{E}= C_{E} \otimes \overline{C}_{E}$ as an operator in $\K$,
 where $\overline{C}_{E}  = \Jj C_{E}\Jj$, with $\Jj$ being complex conjugation: $\Jj f= \bar{f}$ for 
any $f \in \Hh$.  Since $J$ is  anti-unitary on  $\Hh$ we get
\[
\si_{\Hh}(\overline{C}_{E}) =\overline{\si_{\Hh}(C_{E}) }\;, 
\] 
and hence
\begin{equation}
\si_{\K}({\Cc}_{E}) \; = \;\{ \lambda_{J,J'} =\lambda_J\bar{\lambda}_{J'} \,;\;\; J,J'\,\in\,\Delta(m, \ZZ_+)\,  \} 
\cup \{0\}\;,
\end{equation}
with $\lambda_J$ given by \eqref{eq-lbj}.  The same arguments as in the previous Lemma show that 
${\Cc}_{0,E}^2$ is a compact operator on ${\K}_\infty$, so it follows that
\begin{equation}
\si({\Cc}_{E}) := \si_{{\K}_\infty} ({\Cc}_{E})  = \si_{\K}(\Cc_{E})  \;.
\end{equation}
Since $\lambda_{J,J'} \not= \frac1{K}$ for any $ J,J'\,\in\,\Delta(m, \ZZ_+)$, (\ref{000}) follows.  
\end{proof}  

\begin{remark} For the Lemmas~\ref{F} and \ref{Q} it is  crucial that $K\geq 2$. For the one dimensional strip, 
{where $K=1$,} $\lambda_{\mathbf{0}}=1$ (0 matrix for $J$) and $\lambda_{\mathbf{0}}\bar\lambda_{\mathbf{0}}=1$
lead to zero eigenvalues for
$ F_f(0, E,0,\ze_{0,E})$ and $Q_g(0, E,0,\xi_{0,E})$\,. For this reason the proof does not work in the one-dimensional strip.
In fact it is known that in this case one obtains Anderson localization instead of absolutely continuous spectrum
even for small disorder \cite{KLS}.
\end{remark}

We now use the Implicit Function Theorem on Banach Spaces as stated in \cite[Appendix~B]{K2}, a rewriting of \cite[Theorem~2.7.2]{N}.   If $E\in I_{A,K} $, it follows from Lemmas~\ref{F} and \ref{Q} that the hypotheses of this theorem
 are verified for  the functions $F(\lb, E,\eta,f)$ and 
$Q(\lb, E,\eta,g)$ at $(0,E,0,\ze_{0,E} )$ and
 $(0,E,0,\xi_{0,E} )$, respectively.
As a consequence, for each  $E \in I_{A,K}$
there exist $\lb_E > 0$, $\varepsilon_E >0$, $\eta_E >0$ and $\de_E >0$, such that for each 
\[  
(\lb, E',\eta) \in (-\lb_E,\lb_E)\times (E - \ve_E,E + \ve_E) \times [0, \eta_E) 
\]
there is a unique $\, \omega_{\lb,E',\eta} \in {\K}_\infty$ with 
$\| \omega_{\lb,E',\eta}- \xi_{0,E} \|_{ {\K}_\infty}< \de_E$, 
such that we have  $  Q(\lb, E',\eta,\omega_{\lb,E',\eta}) =0$.  Moreover, the map
\[
 (\lb, E',\eta) \in (-\lb_E,\lb_E)\times (E - \ve_E,E + \ve_E) \times [0, \eta_E) \;\longrightarrow\;
 \omega_{\lb,E',\eta} \in {\K}_\infty
\]
is continuous.  Similar statements hold for $F(\lb, E,\eta,f)$.   

%@xize
\begin{theorem}  \label{xize}
For any $E\in I_{A,K}$ there exist $\lb_E > 0$ and $\varepsilon_E >0$, 
 such that the maps 
\begin{equation}
(\lb, E',\eta) \in (-\lb_E,\lb_E)\times (E - \ve_E,E + \ve_E) \times (0, \infty) \;\longrightarrow\;
 \xi_{\lb,E'+i\eta} \in {\K}_\infty  \label{xizexi}
\end{equation}
and 
\begin{equation}
(\lb, E',\eta) \in (-\lb_E,\lb_E)\times (E - \ve_E,E + \ve_E) \times (0, \infty) \;\longrightarrow\;
 \ze_{\lb,E'+i\eta} \in {\Hh}_\infty  \label{xizeze}
\end{equation}
have  continuous extensions to $ (-\lb_E,\lb_E)\times (E - \ve_E,E + \ve_E) \times [0, \infty)$
satisfying \eqref{eq-fpx} and \eqref{eq-fpz}, respectively.   
\end{theorem}

\begin{proof}   For the map given in (\ref{xizexi}) it suffices to prove that 
\begin{equation}
\xi_{\lb,E' + i\eta}\; = \;\omega_{\lb,E',\eta}\;\;\mbox{for all}\;\;
(\lb, E',\eta) \in (-\lb_E,\lb_E)\times (E - \ve_E,E + \ve_E) \times (0, \eta_E)\;.  \label{eq}
\end{equation}
 But it 
follows  from Proposition~\ref{xita} that  $\xi_{\lb,E'+i\eta}$ is a continuous function of
$(\lb, E',\eta)$ in the set
\[\left(\{0\} \times  \{E'\}\times  [0, \eta_1]\right) \cup  
 \left(\RR \times\RR \times [\eta_1, \infty) \right)\,,
\]
for any $\eta_1 >0$, which satisfies 
(\ref{eq-fpx}).  Thus \eqref{eq} follows from the uniqueness in the Implicit Function Theorem.  
The proof for the map in \eqref{xizeze} is similar.  
\end{proof}

%%%%%%%%%%%%%%%%%%%%%%%%%%%%%%%%%%%%%%%%%%%%%%%%%%%%%%%%%%%%%%%%%%%%%%%%%%%%%%%%%%%%%%%%
%%%%%%%%%%%%%%%%%%%%%%%%%%%%%%%%%%%%%%%%%%%%%%%%%%%%%%%%%%%%%%%%%%%%%%%%%%%%%%%%%%%%%%%%
%%%%%%%%%%%%%%%%%%%%%%%%%%%%%%%%%%%%%%%%%%%%%%%%%%%%%%%%%%%%%%%%%%%%%%%%%%%%%%%%%%%%%%%%
%%%%%%%%%%%%%%%%%%%%%%%%%%%%%%%%%%%%%%%%%%%%%%%%%%%%%%%%%%%%%%%%%%%%%%%%%%%%%%%%%%%%%%%%
   
\section{Proofs of the main theorems}\label{sec-proofs}

%%%%%%%%%%%%%%%%%%%%%%%%%%%%%%%%%%%%%%%%%%%%%%%%%%%%%%%%%%%%%%%%%%%%%%%%%%%%%%%%%%%%%%%%
%%%%%%%%%%%%%%%%%%%%%%%%%%%%%%%%%%%%%%%%%%%%%%%%%%%%%%%%%%%%%%%%%%%%%%%%%%%%%%%%%%%%%%%%
%%%%%%%%%%%%%%%%%%%%%%%%%%%%%%%%%%%%%%%%%%%%%%%%%%%%%%%%%%%%%%%%%%%%%%%%%%%%%%%%%%%%%%%%
%%%%%%%%%%%%%%%%%%%%%%%%%%%%%%%%%%%%%%%%%%%%%%%%%%%%%%%%%%%%%%%%%%%%%%%%%%%%%%%%%%%%%%%%

Theorem \ref{cor}(i) now follows from \eqref{eq-EGG}, \eqref{eq-EG}, Theorem~\ref{xize}, the 
translation invariance of expectations, and a simple compactness argument.
Parts (ii) and (iii)  follow from part (i), as explained in the introduction. 

Theorem~\ref{main} follows from Theorem~\ref{cor}(iii).   Let $I=[a,b]\subset I_{A,K} $ and  $\lb(I) > 0$ be as in Theorem~\ref{cor}, 
so \eqref{sup} holds.  For $|\lb| < \lb(I) $ and any $x \in \B$ we use
Fubini's Theorem and Fatou's Lemma to obtain
\begin{align}
&\E \left( \liminf_{\eta \downarrow 0}\int_{a}^b \Tr(|G_{\lb}\, (x,x; E +i\eta)|^2)\,dE \right)\\
& \qquad\qquad\qquad 
\leq\;\; \liminf_{\eta \downarrow 0} \int_{a}^b \E (\Tr(|G_{\lb}\, (x,x; E +i\eta)|^2)  )\,dE\;\;<\;\; \infty\;. \notag
\end{align}
Thus, 
\beq\label{limprobone}
\liminf\limits_{\eta \downarrow 0}\int\limits_{a}^b \Tr(|G_{\lb}\, (x,x; E +i\eta)|^2)\,dE \,<\,\infty \quad
\text{with probability one}. 
\eeq
Let $d\nu_{\lb,x,k}(E)  = \left\langle {x,k|dP_\lb (E)|x,k} \right\rangle$,
where $dP_\lb (E)$ is the spectral measure of the operator $H_\lb$. 
The Stieltjes transform of $d\nu_{\lb,x,k}$ is given by
$(G_{\lb}\, (x,x; E +i\eta))_{k,k}$.
In view of \eq{limprobone}, it follows from \cite[Theorem~4.1]{K2} that, with probability one,  the finite measure  $\nu_{\lb,x,k}|_{(a,b)}$ is purely absolutely continuous  for all
$x \in \B,\,k\in \{1,\ldots,m\}$,
so Theorem~\ref{main} is proved.  (Although \cite[Theorem~4.1]{K2} is stated for intervals of the form $(-a,a)$, it clearly holds for general bounded intervals $(a,b)$.)

%%%%%%%%%%%%%%%%%%%%%%%%%%%%%%%%%%%%%%%%%%%%%%%%%%%%%%%%%%%%%%%%%%%%%%%%%%%%%%%%%%%%%%%%%%%%%%%%%%%%%%
%%%%%%%%%%%%%%%%%%%%%%%%%%%%%%%%%%%%%%%%%%%%%%%%%%%%%%%%%%%%%%%%%%%%%%%%%%%%%%%%%%%%%%%%%%%%%%%%%%%%%%
%%%%%%%%%%%%%%%%%%%%%%%%%%%%%%%%%%%%%%%%%%%%%%%%%%%%%%%%%%%%%%%%%%%%%%%%%%%%%%%%%%%%%%%%%%%%%%%%%%%%%%
%%%%%%%%%%%%%%%%%%%%%%%%%%%%%%%%%%%%%%%%%%%%%%%%%%%%%%%%%%%%%%%%%%%%%%%%%%%%%%%%%%%%%%%%%%%%%%%%%%%%%%

%%%%%%%%%%%%%%%%%%%%%%%%%%%%%%%%%%%%%%%%%%%%%%%%%%%%%%%%%%%%%%%%%%%%%%%%%%%%%%%%%%%%%%%%%%%%%%%%%%%%%%
%%%%%%%%%%%%%%%%%%%%%%%%%%%%%%%%%%%%%%%%%%%%%%%%%%%%%%%%%%%%%%%%%%%%%%%%%%%%%%%%%%%%%%%%%%%%%%%%%%%%%%
%%%%%%%%%%%%%%%%%%%%%%%%%%%%%%%%%%%%%%%%%%%%%%%%%%%%%%%%%%%%%%%%%%%%%%%%%%%%%%%%%%%%%%%%%%%%%%%%%%%%%%
%%%%%%%%%%%%%%%%%%%%%%%%%%%%%%%%%%%%%%%%%%%%%%%%%%%%%%%%%%%%%%%%%%%%%%%%%%%%%%%%%%%%%%%%%%%%%%%%%%%%%%


\begin{thebibliography}{99}
\bibitem
{AALR}  E. Abrahams,  P. Anderson, D. Licciardello and T. Ramakrishnan, \emph{
Scaling theory of localization: absence of quantum diffusion in two dimensions},  Phys. Rev. Lett. \textbf{42}, 673-675 (1979)

\bibitem{AK}  V. Acosta and A. Klein, \emph{ Analyticity of the density of states in the Anderson model
in the Bethe lattice},  J. Stat. Phys.  {\bf 69}, 277-305 (1992)


\bibitem{A}  M. Aizenman, \emph{ Localization at weak disorder:  some elementary bounds},
Rev. Math. Phys. {\bf 6}, 1163-1182 (1994)


\bibitem{ASW} M. Aizenman, R. Sims and S. Warzel, \emph{ Stability of the absolutely continuous spectrum of random
Schr\"odinger operators on tree graphs}, Prob. Theor. Rel. Fields, {\bf 136}, 363-394 (2006)


\bibitem{AM}  M. Aizenman and S. Molchanov, \emph{ Localization at large disorder and extreme
energies:  an elementary derivation},  Commun. Math. Phys. {\bf 157}, 245-278 (1993)

\bibitem{And}  P. Anderson,  \emph{ Absence of diffusion in certain random lattices}, 
Phys. Rev. {\bf 109}, 1492-1505 (1958)

\bibitem{B} F.A. Berezin, \emph{ The method of second quantization}, Academic Press, New York, 1966


\bibitem{CK}  M. Campanino and A. Klein, \emph{ A supersymmetric transfer matrix and 
differentiability of the density of states  in the one-dimensional Anderson model},
Commun. Math. Phys. {\bf 104}, 227-241 (1986)


\bibitem{CKM}  R. Carmona, A. Klein and F. Martinelli,  
\emph{ Anderson localization for Bernoulli and other singular potentials}, 
Commun. Math. Phys. {\bf 108}, 41-66 (1987)

 
\bibitem{CL}  R. Carmona and J. Lacroix, \emph{ Spectral Theory of Random Schrodinger Operators},
  Boston, MA:  Birkhauser, 1990


\bibitem{DLS}  F. Delyon, Y. Levy and B. Souillard,  \emph{ Anderson
localization for multidimensional systems at large disorder or low
energy},  Commun. Math. Phys. {\bf 100}, 463-470 (1985)

\bibitem{Dor} O. N. Dorokhov, 
\emph{ Solvable model of multichannel localization}, Phys. Rev.
B {\bf 37} 10526-10541 (1988)

\bibitem{DK}  H. von Dreifus and A. Klein,  \emph{ A new proof of localization
in the Anderson tight binding model},  Commun. Math. Phys. {\bf 124},
285-299 (1989)

\bibitem{E} K. B. Efetov, 
\emph{ Supersymmetry and theory of disordered metals}, Advances Phys. {\bf 32}, 53-127 (1983)


\bibitem{FHH} R. Froese, F. Halasan and D. Hasler, 
\emph{ Absolutely continuous spectrum for the Anderson model on a product
of a tree with a finite graph}, J. Funct. Anal. {\bf 262}, 1011-1042 (2012)


\bibitem{FHS} R. Froese, D. Hasler and W. Spitzer, 
\emph{ Absolutely continuous spectrum for the Anderson Model on a tree:
A geometric proof of Klein's Theorem}, Commun. Math. Phys., {\bf 269}, 239-257 (2007)


\bibitem{FMSS}  J. Fr\"ohlich, F. Martinelli, E. Scoppola and T. Spencer
 \emph{ Constructive proof of localization in the Anderson tight binding
model},  Commun. Math. Phys. {\bf 101}, 21-46 (1985)


\bibitem{FS}  J. Fr\"ohlich and T. Spencer,  \emph{ Absence of diffusion in the
Anderson tight binding model for large disorder or low energy},
Commun. Math. Phys. {\bf 88}, 151-184 (1983)


\bibitem{GMP}  Ya. Gol'dsheid, S. Molchanov and L. Pastur,  
\emph{ Pure point spectrum of stochastic one dimensional Schr\"odinger operators},
Funct. Anal. Appl. {\bf 11}, 1-10 (1977)


\bibitem{H} F. Halasan, \emph{ Absolutely continuous spectrum for the Anderson model on some tree-like graphs}, 
arXiv:0810.2516v3 (2008)



\bibitem{Ka}  T. Kato,  \emph{ Wave operators and similarity for some non self-adjoint operators},
Mat. Ann. {\bf 162}, 258-279 (1966) 


\bibitem{KLW} M. Keller, D. Lenz and S. Warzel, 
\emph{ On the spectral theory of trees with finite cone type}, arXiv:1001.3600v2 (2011),
Israel J. Math., to appear


\bibitem{K}  A. Klein,  \emph{ The supersymmetric replica trick and smoothness of the density of states 
for random Schrodinger operators},   Proc. Symposia in Pure Mathematics {\bf 51}, 315-331 
(1990)

\bibitem{K3}  A. Klein,  \emph{ Localization in the Anderson model with long range hopping}, 
Braz. J. Phys. {\bf 23}, 363-371 (1993)

\bibitem{K1}  A. Klein, \emph{ Absolutely continuous spectrum in the Anderson model on the Bethe lattice},
  Math. Res. Lett. {\bf 1}, 399-407 (1994)
  
 \bibitem{K9}  A. Klein, \emph{ Absolutely continuous spectrum in random Schr\"odinger operators},
 Quantization, nonlinear partial differential equations, and operator
 algebra (Cambridge, MA, 1994), 
 139-147, Proc. Sympos. Pure Math. \textbf{59}, Amer. Math. Soc., Providence, RI,  1996
 
 \bibitem{K8}  A. Klein,   \emph{ Spreading of wave packets in the Anderson model on the Bethe lattice},
Commun. Math. Phys. {\bf 177}, 755--773 (1996)
  
\bibitem{K2} A. Klein, \emph{ Extended states in the Anderson model on the Bethe lattice},
Advances in Math. {\bf 133}, 163-184 (1998)

 \bibitem{KLS2}    A. Klein, J. Lacroix and A. Speis, \emph{ Regularity of the density of states in the Anderson model on a strip
 for potentials with singular continuous distributions},
 J. Statist. Phys.  {\bf 57},   65-88  (1989)

 \bibitem{KLS}    A. Klein, J. Lacroix and A. Speis,  \emph{ Localization for the
Anderson model on a strip with singular potentials},  J. Funct. Anal.
 {\bf  94}, 135-155 (1990)

\bibitem{KS} A. Klein and C. Sadel,
\emph{Ballistic Behavior for Random Schr\"odinger Operators on the Bethe Strip},
J.~Spectr. Theory {\bf 1}, 409-442 (2011)

\bibitem{KSp2} A. Klein and A. Speis,  \emph{ Smoothness of the density of states in
the Anderson model on a one-dimensional strip}, Annals of Phys. {\bf 183}, 352-398 (1988)

\bibitem{Klo} F. Klopp,  \emph{ Weak disorder localization and Lifshitz tails},  
Commun. Math. Phys. \textbf{232}, 125-155 (2002)


\bibitem{KS1}  H. Kunz and B. Souillard,  \emph{ Sur le spectre des operateurs
aux differences finies aleatoires},  Commun. Math. Phys. {\bf 78},
201-246 (1980)


\bibitem{Lac}  J. Lacroix, \emph{ Localisation pour l'op\'erateur de
Schr\"odinger al\'eatoire dans un ruban},  Ann. Inst. H. Poincar\'e ser
{\bf A40}, 97-116 (1984)


\bibitem{N}  L. Nirenberg,  \emph{ Topics in Nonlinear Functional Analysis},  New York:  Courant Institute of 
Mathematical Sciences, 1974    

\bibitem{P}   L. Pastur,  \emph{ Spectra of random selfadjoint operators}  Russ. Math. Surv. 
 {\bf 28}, 1-67 (1973)


\bibitem{SS} C. Sadel and H. Schulz-Baldes, 
\emph{ Random Dirac Operators with time reversal symmetry},
Commun. Math. Phys. {\bf 295}, 209-242 (2010)


\bibitem{SW}  B. Simon and T. Wolff, \emph{ Singular continuum spectrum under
rank one perturbations and localization for random Hamiltonians},
Commun. Pure. Appl. Math. {\bf 39}, 75-90 (1986)

\bibitem{Weg} F. Wegner, \emph{ Disordered systems with $n$ orbitals per site: $n=\infty$ limit},
Phys. Rev. B. {\bf 19}, 783-792 (1979)

\bibitem{Wang} W.-M Wang, \emph{ Localization and universality of Poisson statistics
 for the multidimensional Anderson model at weak disorder}, 
 Invent. Math. \textbf{146}, 365-398 (2001)
  

\end{thebibliography}
\end{document}